\documentclass[12pt,final,journal,oneside,onecolumn]{IEEEtran}

\usepackage{amsmath}
\usepackage{amsthm}
\usepackage{graphicx}
\usepackage{amsfonts}
\usepackage{multirow}
\usepackage{subfigure}
\usepackage{algorithm2e}
\usepackage{pdfpages}
\usepackage{longtable}
\usepackage{ucs}
\usepackage[utf8x]{inputenc}
\usepackage{url}

\newcommand{\Hhat}{\mathbf{\hat{H}}}
\newcommand{\Hhatt}{\mathbf{\hat{H}}^\text{T}}
\newcommand{\PP}{\mathbf{P}}
\newcommand{\PPt}{\mathbf{P}^\text{T}}
\newcommand{\sig}{\mathbf{s}}
\newcommand{\sigt}{\mathbf{s}^\text{T}}
\newcommand{\yy}{\mathbf{y}}
\newcommand{\yyt}{\mathbf{y}^\text{T}}
\newcommand{\SIGMAs}{\Sigma_{\sig}}
\newcommand{\mus}{\mathbf{\mu_s}}
\newcommand{\musi}{\mathbf{\mu}^i_{\mathbf{s}}}
\newcommand{\must}{\mathbf{\mu_s}^\text{T}}
\newcommand{\musti}{{\mathbf{\mu}^i_\mathbf{s}}^\text{T}}
\newcommand{\argmin}[1]{\underset{#1}{\operatorname{argmin}} \hspace{0.15cm}}

\newcommand{\mc}{\mathcal}
\newcommand{\mb}{\mathbf}
\newcommand{\mylb}{\left[}
\newcommand{\mylp}{\left(}
\newcommand{\myrb}{\right]}
\newcommand{\myrp}{\right)}

\theoremstyle{plain}
\newtheorem{proposition}{Proposition}

\theoremstyle{plain}

\theoremstyle{definition}
\newtheorem{definition}{Definition}

\theoremstyle{remark}
\newtheorem{remark}{Remark}

\theoremstyle{definition}
\newtheorem{solution}{Solution}

\begin{document}

\title{\textit{Channel Whispering}: a Protocol for \\Physical Layer Group Key Generation \\
\begin{Large}
Application to IR-UWB through Deconvolution
\end{Large} 
}

\author{
\IEEEauthorblockN{I. Tunaru, B. Denis \IEEEauthorrefmark{1}, R. Perrier \IEEEauthorrefmark{1}, B. Uguen \IEEEauthorrefmark{2} \\}
\IEEEauthorblockA{\IEEEauthorrefmark{1}
CEA-Leti Minatec, GIANT Campus, Grenoble, France \\
}
\IEEEauthorblockA{\IEEEauthorrefmark{2}
IETR, University of Rennes 1, Rennes, France \\
Email: iulia.tunaru.it@gmail.com, \{benoit.denis,regis.perrier\}@cea.fr, bernard.uguen@univ-rennes1.fr
}}

\maketitle

\begin{abstract}
As wireless \textit{ad hoc} and mobile networks are emerging and the transferred data become more sensitive, information security measures should make use of all the available contextual resources to secure information flows.  
The physical layer security framework provides models, algorithms, and proofs of concept for generating pairwise symmetric keys over single links between two nodes within communication range. In this study, we focus on cooperative group key generation over multiple Impulse Radio - Ultra Wideband (IR-UWB) channels according to the source model. 
The main idea, proposed in previous work, consists in generating receiver-specific signals, also called \textit{s-signals}, so that only the intended receiver has access to the non-observable channels corresponding to its non-adjacent links. Herein, we complete the analysis of the proposed protocol and investigate several signal processing algorithms to generate the \textit{s-signal} expressed as a solution to a deconvolution problem in the case of IR-UWB. Our findings indicate that it is compulsory to add a parameterizable constraint to the searched \textit{s-signal} and that the Expectation-Maximization algorithm can provide a stable self-parameterizable solution. Compared to physical layer key distribution methods, the proposed key generation protocol requires less traffic overhead for small cooperative groups while being robust at medium and high signal-to-noise ratios. 
\end{abstract}


\IEEEpeerreviewmaketitle

\section{Introduction}
\label{sec:intro}

A considerable amount of data is generated, exchanged, and collected in modern device-centric applications such as context-aware services in Smart Cities, 
nomadic socialized Internet of Things (IoT), advanced human/machine-to-machine interfaces and communications, participatory sensing and environment monitoring for Big Data analysis, local multi-agent cooperative data fusion, Cooperative Intelligent Transportation Systems (C-ITS). These applications will require emerging networks that have to carry confidential information between mobile users in \textit{ad-hoc} scenarios or from remote sensors to a core-network. 

Most wireless networks capable of supporting such applications naturally require peer-to-peer interactions between end-devices under opportunistic connectivity conditions. This type of scenarios are foreseen by WiFi Direct and Device-to-Device (D2D) options in pending 5G standards, IEEE 802.11p-compliant Vehicular Ad hoc Networks (VANETs), and short range technologies such as Near Field Communications (NFC), Bluetooth-Low Energy (BT-LE), IEEE 802.15.4 (Zigbee), IEEE 802.15.4a or IEEE 802.15.6 Impulse Radio - Ultra Wideband (IR-UWB). Highly susceptible to eavesdropping and impersonation attacks by nature, these emerging wireless networks might also be subject to hardly predictable mobility patterns, erratic users' activity, and varying devices densities, hence requiring flexible and scalable security measures. 

Currently, communication systems rely mainly on high layer symmetric cryptography (based on common secret keys) for data encryption/decryption and on Public Key Cryptography (PKC) for authentication and symmetric key distribution \cite{Katz14}. However, in wireless decentralized or \textit{ad hoc} networks, symmetric key distribution with PKC is deemed challenging. Besides requiring high computational complexity from both implementation and execution perspectives, PKC needs a centralized certified management entity to distribute, refresh and revoke keys or signatures. The latter might be problematic when mobile devices erratically associate or leave the local network assuming short-range physical connectivity. 

Alternative approaches for symmetric key management in decentralized networks include: \textit{i)} the Diffie-Hellman (DH) protocol for pairwise key distribution or its extensions to group key distribution \cite{Bresson2001}; \textit{ii)} low complexity variants of PKC with preliminary secure distribution of location-based private keys \cite{Zhang06_manet}; \textit{iii)} pre-distribution of keying materials at the devices of interest in the deployment phase. It can be noticed that these options require preliminary keying material for generating group keys and become therefore challenging in \textit{ad hoc} or mobile scenarios with opportunistic connections.  

Independently or in support of the aforementioned methods, the physical layer key generation paradigm (PLKG) \cite{Shehadeh15} and its different models (i.e., source, channel, mixed \cite{TunaruPHD}) have been recently put forward starting from information-theoretic studies of secret sharing \cite{Csiszar93} \cite{Bloch11}. Furthermore, some of the PLKG models have been adapted to cooperative scenarios involving several nodes, either to reinforce the generated pairwise keys or to issue a common group key (i.e., shared by more than two nodes).

In the source model of PLKG, on which we focus in this paper, pairwise symmetric keys can be generated after bilaterally measuring the wireless random\footnote{The physical channel is not truly random but it is perceived as random because of the prohibitive complexity of exact deterministic signal reconstruction based on the precise environment characterization especially for IR-UWB (e.g., by ray tracing methods \cite{Raytracing}).} channel between two concerned nodes, quantizing it, and correcting the possible errors. The operation of channel measurement between a transmitter and a receiver is also known as channel probing or channel sounding. In this paper, we investigate cooperative group key generation from IR-UWB multipath channels according to the source model. In previous work \cite{TunaruICUWB15}, we proposed a new physical layer protocol to generate group keys within cooperative scenarios while exploiting all the available physical links in a small full mesh topology and reducing over-the-air traffic with respect to other cooperative higher layer PLKG methods. The main idea consists in adjusting the IR-UWB signals usually transmitted for channel probing so that a target node -and only this node-, can have access to non-observable channels corresponding to its non-adjacent links.\footnote{Non-adjacent links are links between other nodes in the network.} In this way, a collaborative node can \textit{discreetly whisper a channel} to one of its neighbors that does not have physical access to it. This operation leads to a deconvolution problem in the case of IR-UWB, for which we now discuss various solutions that have not been investigated in the initial work \cite{TunaruICUWB15} \cite{Tunaru2017patent}. Firstly, we show that the maximum-likelihood (ML) approaches to the deconvolution problem are not stable with respect to imperfect channel impulse response (CIR) estimates. Then, we introduce a parametrized maximum \textit{a posteriori} (MAP) solution and we analyze two automatic methods for parameterization: Cross Validation (CV) and Expectation Maximization (EM).


The paper starts with an overview of pairwise PLKG followed by a description of the state of the art cooperative key generation methods and a synopsis of our contributions and the limitations of our work (Section \ref{sec:sota}). 
In Section \ref{sec:system} we introduce the system model and definitions. Our main contributions are detailed in Section \ref{sec:methods} starting with the protocol description (Section \ref{sec:protocol}) and continuing with the deconvolution options to compute the optimized transmitted signal (Section \ref{sec:deconvolution}). In Section \ref{sec:results} the deconvolution solutions are compared and the proposed protocol is analyzed in terms of traffic complexity, key length gains, and average bit matching. Finally, section \ref{sec:conclusion} concludes the paper and discusses the perspectives.

\section{Related work}
\label{sec:sota}

\subsection{Single-link physical layer key generation}
\label{sec:single_sota}

Secret key generation based on the physical layer in wireless communications is a particular case of information theoretic secret generation, a more general framework that consists of two main models: the source and the channel models \cite{Bloch11}. In the following we will only describe the source model, to which our work belongs.

Radio propagation characteristics can be used as common source of information for secret key agreement \cite{Badawy16_general}. Due to the unpredictable fading realizations and to the reciprocity of the propagation of electromagnetic waves, the wireless channel between two legitimate users represents a common source of randomness that can be exploited to separately generate a secret key and agree on it through public discussion. A common assumption is that any eavesdropper, situated in a sufficiently distant position with respect to the legitimate users, observes an uncorrelated channel and, therefore, will not be able to generate the same key. A typical point-to-point sequential key generation algorithm consists of \cite{Bloch11}:
\begin{itemize} 
\item \emph{Randomness sharing} (e.g., from channel probing).
\item \emph{Advantage distillation}, an optional step that aims at selecting the channel probes for which the legitimate users consider to have an advantage with respect to an eavesdropper (not needed in the considered source model thanks to the reciprocity and the spatial decorrelation properties). 
\item \emph{Information reconciliation} meant to correct the mismatches due to asymmetric equipment, noise, interferences and temporally distant half-duplex communications by using exchanges over a public channel; this step is usually preceded or jointly implemented with a \emph{quantization} phase, which transforms values issued from channel measurements into binary flows.
\item \emph{Privacy amplification}: deterministic processing of the common bit sequences in order to generate a secure secret-key by ``compensating'' the information leakage on the public channel (e.g., hash functions or randomness extractors).
\end{itemize}

The trade-offs arising in single-link physical layer key generation have been studied for different technologies and channel metrics: received signal strength (RSS) \cite{Patwari10}, OFDM coefficients \cite{Alfandi12} \cite{Badawy16}, IR-UWB channel responses \cite{Paolini14} \cite{TunaruPIMRC14} \cite{Huang15Wiley}.  

The aforementioned methods provide a symmetric key shared by only two nodes (i.e., point-to-point) and generated from a single communication channel (i.e., single link). The channel can be also probed in time (over successive transmissions) for obtaining longer keys or for refreshing keys but the performance is dependent on the channel coherence time (e.g., in mobility cases). 

Therefore, an elementary issue for physical layer key generation is how to gather more entropy from channel measurements during a limited amount of time when the channel can be considered static. This leads to the idea of extending the key generation process to several nodes in order to exploit more physical links (cooperative/multi-link strategy) and/or to generate a group key (contrary to a point-to-point key). However, the solution should be scalable, adapted to \textit{ad hoc} scenarios and should avoid the high level complexity issues of classical key distribution techniques (e.g., key pre-distribution, latency, etc.).

\subsection{Multi-link physical layer key generation}
\label{sec:multi_sota}

The secret key capacities for the source model with multiple terminals including a subset of helpers, various extents of an eavesdropper's knowledge, and unrestricted public discussion have been characterized from an information theoretic perspective \cite{Csiszar04}. First, it has been shown that the secret key capacity when the eavesdropper only observes the public exchanges without having any side information regarding the source is closely related to the multiterminal source coding problem with no secrecy constraints. Then, the expression of the secret key capacity when the eavesdropper wiretaps a subset of the helpers has been also derived. 

Multi-terminal or cooperative secret key generation has also been investigated for less complex systems with the aim to design practical protocols and to measure their performance. One early study on the topic presents an extension of the source model to cooperative pairwise key agreement and group key generation in a pairwise independent network (i.e., a network in which the point-to-point channels are independent) \cite{Ye07}: point-to-point keys are generated from each physical network link and the group keys (or extra secret bits for pairwise cooperative keys) are propagated through XOR-ing operations over a graph representation of the network. 

Other authors propose to generate the secret key between two user nodes with the help of a relay \cite{Liang12} \cite{Wang12}. First, non-cooperative pairwise keys (from the main channel and the side channels, i.e., the channels between each node and the relay) are generated using a typical key distillation procedure based on channel gains \cite{Liang12} or phase estimations \cite{Wang12}. At this point, each user node holds the key from the main channel and a key from a side channel. After that, the relay, which holds the keys obtained from both side channels, publicly broadcasts the bitwise combination of those keys. The user nodes can therefore recover the key that was generated by the other user node and the relay and append one of these side channel keys to the main channel key. An equivalent approach based on pairwise keys is described in parallel work \cite{Wei12}, where the received signal strength is quantized for non-cooperative pairwise key generation and a group key, generated by a root node, is securely distributed in the network using the pairwise keys. 

If the wireless channels are static in time or too sparse in terms of multipath, there might not be enough information to harvest in order to generate a robust secret key. In order to deal with the issue of limited entropy of the source model, a more recent study \cite{Wang14} extends the channel model for key generation to a cooperative scenario with a relay and an eavesdropper that is collocated with the relay. The authors derive the upper and lower bound for the secret key rate with a relay and propose a joint optimized design of the various key generation phases (advantage distillation, information reconciliation and privacy amplification) while focusing on the trade-off between security and protocol efficiency. Although it is shown that collocating with the relay is the worst case scenario for the secret key rate, this assumption also facilitates the advantage distillation phase because the legitimate users are informed about the quality of the eavesdropper's signal by the cooperative relay. 

\subsection{Contributions and limitations}
\label{sec:contributions}

Overall, the aforementioned work on cooperative pairwise/group keys rely on the initial pairwise single-link key generation and subsequent key distribution. This involves extra-traffic and latency, while the length of the group key can be limited by the shortest pairwise keys. Herein, we describe a method of physical layer group key generation that avoids the pairwise key generation before group key generation. An alternative way to deal with the entropy limitation in the case of the source model would be to use multiple links as input signals for quantization. The final group key is obtained after quantization and reconciliation of a concatenation of measurements from several links. The protocol involves several cooperative nodes and the obtained key is by construction known to all the participants, so it becomes a group key. This solution avoids extra public communication overhead but comes with an expense on the signal processing side where supplementary operations for channel probing are needed.

The indirect probing of a non-adjacent channel is possible because of the cooperation between nodes, which will send specific signals (called \textit{s-signals} in the following) in order to \textit{whisper} a certain channel state to the receiver. This operation depends on the type of measured signal; in the case of IR-UWB channel responses, which we will use for illustration purposes, the operation consists in a deconvolution operation, which is one of the scenarios with the highest complexity. This could be nonetheless acceptable for current personal devices such as smartphones and possibly even for next generation wireless sensors. The general concept and the protocol are however applicable to different technologies and channel measurements, which would require less complex signal processing capabilities (e.g., quotient operations). 

It should be noted that similar concepts (e.g., IR-UWB time-reversal \cite{DeNardis11}) have been put forward for improved communication robustness and intrinsic signal secrecy by spatial focusing of the signal energy. These methods rely on pre-filtering on the transmitter side and thus, enable location-dependent SNR gains on the receiver side. However, they are neither intended to provide secret material, such as keys, to higher layer cryptographic functions nor used in cooperative protocols as described in our method.

It has been brought to our knowledge that ideas similar to our \textit{whispering} concept have been studied before \cite{Koksal2014} or in parallel \cite{Quek2015} to our initial work \cite{TunaruICUWB15} \cite{Tunaru2017patent}. In a sense, \textit{channel whispering} is more general than the aforementioned studies because they deal with the RSS metric \cite{Koksal2014} or narrow-band frequency-flat fading channels \cite{Quek2015} for which \textit{the whispering} is equivalent to subtraction and quotient operations respectively. Moreover, the group key generated for the star and chain topologies \cite{Koksal2014} depends only on one shared physical channel. This limitation is surpassed in \cite{Quek2015} where a more general model than ours is given for group key generation in a mesh network (selection of the used channels for key generation, \textit{whispering} of a weighted combination of different estimated channels with optimized coefficients, different quantization schemes including vector quantization). Also, a similar approach is taken in recent work but from an attacker point of view \cite{Harshan2017}.

\paragraph{Contributions}
To sum up, the main contributions of our work on cooperative PLKG are:
\begin{itemize}
\item a cooperative protocol to generate a secret group key within the source model of the paradigm of physical layer key generation \cite{TunaruICUWB15} \cite{Tunaru2017patent}; 
\item the evaluation of several temporal deconvolution solutions required for the application of the cooperative protocol to the IR-UWB physical layer (one of the most challenging for the cooperative key generation purpose), thus assessing the feasibility and the robustness of the solution;
\item a signal model for EM-based deconvolution that can be extended to include a channel estimation error model;
\item a complete comparison of the proposed group key generation protocol to an alternative approach based on group key distribution. This comparison extends the previous results \cite{TunaruICUWB15} by adding a general traffic analysis for an arbitrary number of nodes in the network, emphasizing thus the scalability aspects.
\end{itemize}

\paragraph{Limitations}
The present work focuses on the signal-processing aspects of the achievement of common randomness in IR-UWB mesh networks. However, it could be improved by analyzing the following ``more'' information-theoretic or experimental aspects:
\begin{itemize}
\item leakage analysis: although the channel of the eavesdropper is considered uncorrelated with the legitimate channel, the eavesdropper can recover the \textit{s-signal} by deconvolution. Despite being highly insufficient to deduce the \textit{whispered} channel, this information leakage should be quantified and used in the design of the privacy amplification scheme. 
\item scaling of the recovery noise with the number of nodes: in our work we present the scaling of the traffic packets with the number of nodes in a mesh network but the key rate is hindered by the deconvolution noise which also scales with the number of nodes. This aspect could be incorporated in a general analysis of the secret key rate.
\item realistic measurements and synchronization issues: our study could be completed by using real IR-UWB measurements from small mesh networks or ray-tracing generated signals and artificial noise \cite{Raytracing}. Also, the synchronization issues, which have not been discussed here as explained further (Section \ref{sec:results1}), can have a great impact of the overall performance of the algorithm. This is a common issue to all key generation scheme \cite{Pasolini2015} but its impact could be different in cooperative scenarios and depending on the processing of the received signal before key generation. 
\end{itemize} 

\section{System model}
\label{sec:system}

We consider a full mesh topology consisting of three nodes ($A$, $B$, and $C$) with direct IR-UWB links between each pair. The received signal can be expressed as
\begin{eqnarray}
\label{eq:sounding}
y_{uv}(t) & = & (s_u*h_{uv})(t) + w_{v}(t), \text{ with } 
w_{v}(t) \sim \mathcal{N}(0,\sigma^2_{w}) \\
\label{eq:cir_coop}
h_{uv}(t) & = & \sum_{k=1}^{K}{x_{k} \delta(t-\tau_{k})}
\end{eqnarray}
where $y_{uv}(t)$ of duration $T_w$ is a general convolved noisy received signal between the transmitter $u \in \{A,B,C\}$ and the receiver $v \neq u, v \in \{A,B,C\}$ (later the convolved signal will be denoted as $y$ or $r$ depending on the transmitted signal), $s_u(t)$ is a general signal transmitted by node $u$ (e.g., a pulse waveform $p(t)$ of temporal support $[0,T_p]$ employed for channel sounding or another signal used for channel \textit{whispering}), $h_{uv}$ is the channel impulse response between $u$ and $v$, $x_k$ and $\tau_k$ are respectively the amplitude and delay associated with the $k$\textsuperscript{th} ($k \in \{1,..K\}$) multipath component of the CIR between $u$ and $v$, $w_{v}(t)$ is the additive white Gaussian noise (AWGN) at the receiver $v$.

Our key generation protocol employs all the three available channels in order to generate a secret group key between the three nodes. In the following, we give some definitions and remarks concerning the present system model. 

\begin{definition}
We arbitrarily define the signal-to-noise ratio (SNR) as the ratio between the power of the transmitted pulse over its temporal definition domain and the noise power, assuming that all the CIRs are normalized in energy:\footnote{$\int_0^{T_w} h^2(t) \, \mathrm{d}t = 1$.}
\begin{eqnarray}
\text{SNR}&=\frac{P_{pulse}}{P_{noise}} = \frac{\frac{1}{T_p} \int_0^{T_p} p^2(t) \, \mathrm{d}t}{\sigma^2_{w}}
\end{eqnarray}
\end{definition}

\begin{definition}
An \textit{adjacent channel} is the channel that can be directly measured by a node (e.g., channels $[A-C]$ and $[B-C]$ for node $C$). Accordingly, a \textit{non-adjacent channel} is a channel that cannot be directly measured by a node (e.g., channel $[A-B]$ for node C).
\end{definition}
 
\begin{definition}
We denote as \textit{channel probe} the necessary operations in order for all the nodes to acquire the needed signals for key generation (i.e, adjacent and non-adjacent channels). 
\end{definition}
 
\begin{remark}[Reciprocity and spatial decorrelation]
\label{rem:rec_decorr}
The channels ($h_{AB}$, $h_{AC}$, and $h_{BC}$) are considered reciprocal ($h_{uv}=h_{vu}$) and pairwise independent at a given channel probing time. Although the assumption on the CIR reciprocity is admitted in static environments, some distortions can be induced for the convolved CIRs by the realistic radios with asymmetric amplifiers, filters etc. or by half-duplex channel sounding of each channel. Nevertheless, the reciprocity, the temporal correlation at small temporal scales despite half-duplex sounding, and the spatial decorrelation between different channels are supported by previous experimental studies \cite{Sana} \cite{Madiseh09} \cite{Paolini14}, which have shown that IR-UWB channels are sufficiently reciprocal and spatially uncorrelated when the three nodes are reasonably distant. However, spatial decorrelation is difficult to assess in a general manner because it is directly impacted by the physical environment: the higher the complexity of the physical environment, the more diverse the multipath propagation profiles, and thus the lower the spatial correlation. For our tests, we have employed independent statistical IR-UWB channel realizations based on modified versions of the Saleh-Valenzuela model, namely the IEEE 802.15.4a channel models \cite{ieee802154a}. 
\end{remark}

\begin{remark}[Temporal independence]
\label{rem:independency}
It is assumed that the channels vary from one channel probe to the next one independently of the past channel realizations. The proposed group key generation protocol must be performed only once within the coherence period of the physical channels (i.e., the time over over which the physical channels are static). Depending on the channel probing operations, some protocol variants (e.g., EM-based deconvolution) could need longer coherence times, while others (e.g., MAP-based deconvolution) would be more adapted to dynamic environments. 
\end{remark}

\section{Cooperative physical layer key generation}
\label{sec:methods}

We start by illustrating the proposed cooperative key generation (CKG) algorithm in a group of three interconnected nodes (Section \ref{sec:protocol}). Since the channel probing step will require a deconvolution operation, we analyze several options to achieve this (Section \ref{sec:deconvolution}). For readability reasons, we employ a continuous time framework for the protocol description and we switch to the realistic discrete case when discussing the deconvolution.

\subsection{Protocol description}
\label{sec:protocol}

In the following, $A$, $B$, and $C$ have different roles and the protocol must be repeated with interchanged roles in order to obtain a group key. For a given protocol cycle, we define $A$ as the cooperator, $B$ as the initiator and $C$ as the generator. It is assumed that the radio transmissions between each transmitter and each receiver are preliminarily synchronized at the signal level (e.g., using specific heading). The CKG protocol consists of several phases representing adjacent channel probing, non-adjacent channel \textit{whispering}, and key generation:

\begin{figure}[h]
\centering
\includegraphics{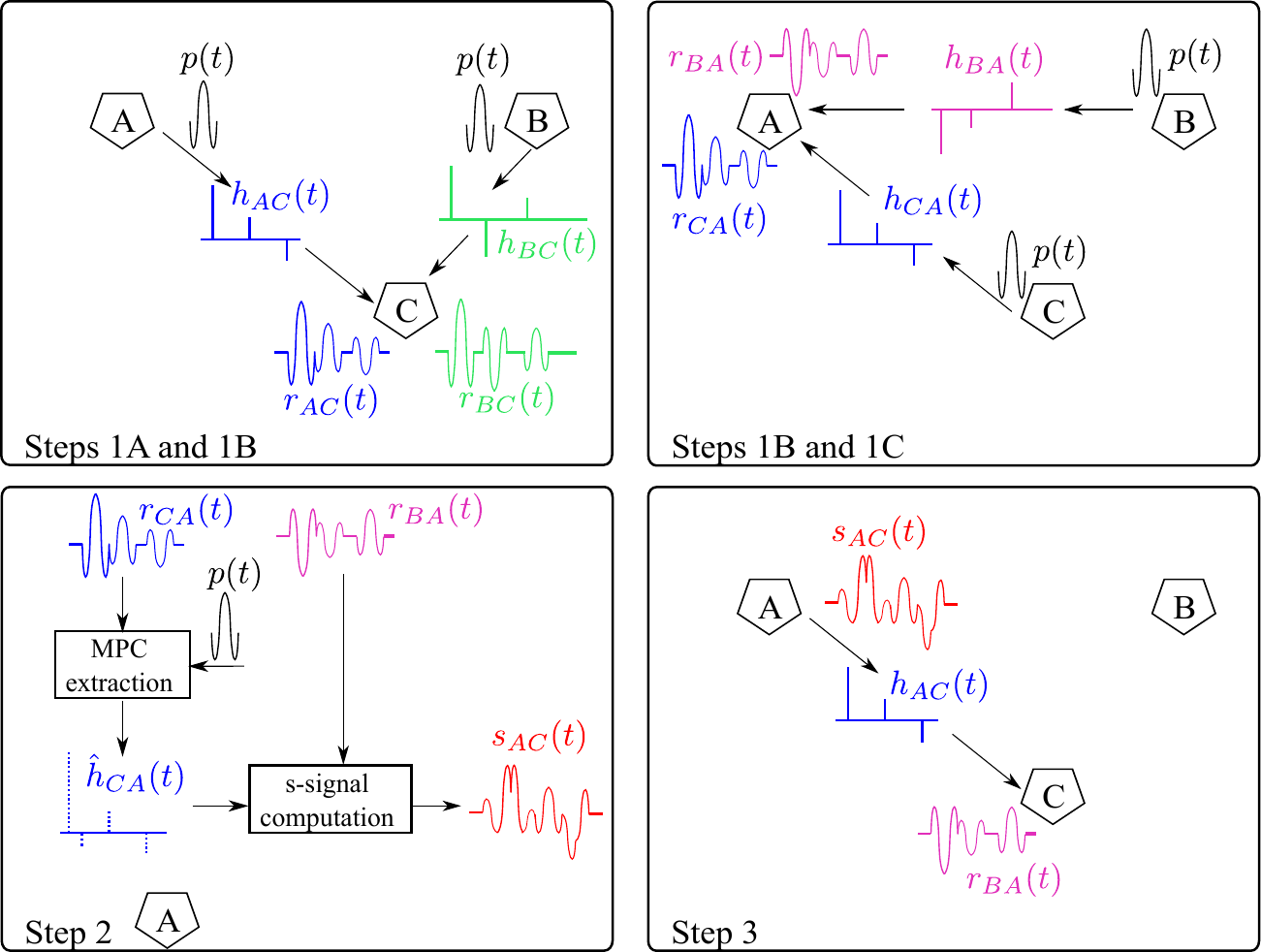}
\caption{Physical layer cooperative key generation between three nodes.}
\label{fig:ckg}
\end{figure}

\begin{itemize}

\item Adjacent channel probing: channel sounding using a pulse template signal to obtain adjacent channel responses for the generator $C$ ($y_{AC}$ and $y_{BC}$) and for the cooperator $A$ ($y_{BA}$ and $y_{CA}$) according to Eq. \eqref{eq:sounding}. This phase can be achieved with three successive broadcast transmissions from each node denoted as 1A, 1B, and 1C in Figure \ref{fig:ckg}.    

\item Non-adjacent channel \textit{whispering}:
\begin{itemize}
\item CIR estimation (i.e., delay and amplitude extraction) at the cooperator $A$ from the received $y_{CA}$ (Step 2 in Figure \ref{fig:ckg}). This operation can be achieved by sampling at frequency $F_s$ and correlation with an \textit{a priori} pulse template (corresponding to the expected unitary received pulse waveform $p(t)$). For example, the search-subtract-readjust algorithm \cite{Gifford11} employed herein iteratively detects, estimates and then subtracts multipath components out of the acquired received signal $y_{CA}$ to obtain $\hat{h}_{CA}$.
	
\item computation of the so-called \textit{s-signal} $s_{AC}$ at the cooperator $A$ (Step 2 in Figure \ref{fig:ckg}). The previous channel estimate $\hat{h}_{CA}$ is used by the cooperator $A$ to compute $s_{AC}$ to be sent to the generator $C$ so that the latter deduces out of its received signal only the non-adjacent channel $[B-A]$ represented by $y_{BA}$.\footnote{The received signal at the generator is ideally close to the non-adjacent channel of the generator.} The problem needs to be solved at the cooperator $A$ is: 
\begin{eqnarray}
\text{Find $s_{AC}$ s.t. } (s_{AC}*\hat{h}_{CA})(t) = y_{BA}(t)
\label{eq:ckg_deconv}
\end{eqnarray}

\item transmission of the computed \textit{s-signal} $s_{AC}$ from $A$ to $C$ (Step 3 in Figure \ref{fig:ckg}). Accordingly, the employed transmitter must enable the programming of an arbitrary IR-UWB waveform, given an \textit{a priori} occupied bandwidth and a power spectral density limitation (e.g., \cite{TCR}), in order to generate $s_{AC}$. The generator $C$ receives:
\begin{equation}
r_{BA}(t) = (s_{AC}*h_{AC})(t) + w_{AC}(t)
\end{equation}

Steps 2-3 can be performed having as cooperator both $A$ and $B$ because they can both measure $C$'s non-adjacent channel $[B-A] \approx [A-B]$. If both of them send an \textit{s-signal}, supplementary processing can be used at $C$ to coherently take advantage of the two incoming signals $r_{BA}$ and $r_{AB}$ in order to obtain a more reliable version of the non-adjacent channel. This extension is not discussed in the present work. Further, steps 2-3 are repeated by all the nodes with interchanged roles so that all of them obtain the non-adjacent channels (i.e., we simply perform three circular rotations). 
\end{itemize}

\item Key generation:
\begin{itemize}
\item processing of all the acquired signals at the generator $C$ (i.e., two direct adjacent channel observations and one non-adjacent channel reconstruction). At this point, $C$ has the signals corresponding to all three channels. These signals are further processed (e.g., through windowing for removing irrelevant parts of the signal, signal squaring, low-pass filtering and down-sampling at frequency $F_p$, compatible with the multipath resolution capability allowed by the signal bandwidth). This processing is preferable in order to obtain signals whose quantization is more robust with respect to mismatches. The generator $C$ then concatenates the three signals to obtain the input quantization signal $S_{CAB}$. The concatenation order can be defined arbitrarily, according to the node ID or based on radio parameters that can be extracted from the three radio signals (e.g., delay spread, mean excess delay or kurtosis as long as these macroscopic radio features are sufficiently reproducible at all nodes).

\item quantization of the input signal $S_{CAB}$ using, e.g., uniform quantization with guard-bands.

\item public discussion between the three nodes: i) sharing the indexes of the dropped samples falling into the guard-bands; ii) error correction using Reed-Solomon codes (a lead node, for example $A$, generates a syndrome representing its own bit sequence and sends it over the public channel to the other nodes, which will try to decode/correct their bit sequences to align them to $A$'s sequence).  
\end{itemize}

\end{itemize}

\subsection{Deconvolution}
\label{sec:deconvolution}

Searching for a solution to Eq. \eqref{eq:ckg_deconv} is a non-trivial deconvolution problem. We restrict our analysis to the temporal domain of IR-UWB signals in order to avoid supplementary processing incurred by the Fourier Transform, but also because of the richness of data processing techniques concerning deconvolution in similar domains (e.g., statistical spatial methods for image deconvolution). We will first present the problem formulation and then explore several solutions. 

\subsubsection{Problem formulation}
\label{sec:formulation}

As mentioned before, the cooperator needs to solve Eq. \eqref{eq:ckg_deconv}, which is now rewritten in discrete form without the node indexes for readability purposes:
\begin{eqnarray}
\Hhat \sig &= \yy
\label{eq:ckg_deconv_matrix}
\end{eqnarray} 
where $\yy$ is the $N \times 1$ sampled version of the generator's non-adjacent channel $y_{BA}(t)$, $\sig$ is the searched \textit{s-signal} of $N_s$ samples, $\Hhat$ is the $N \times N_s$ matrix corresponding to the $N_h \times 1$ convolution kernel $\mb{\hat{h}}$ such that $\Hhat \sig = \mb{\hat{h}} * \sig$. The sampling frequency of these signals is denoted as $F_s$, the same frequency being used for the previous CIR estimation of $\mb{\hat{h}}$.  

\begin{proposition}
In order to ensure the existence of a solution to the deconvolution problem in Eq. \eqref{eq:ckg_deconv_matrix}, the number of samples of the \textit{s-signal} should be $N_s = N+N_h-1$.
\end{proposition}
\begin{proof}
Considering the classical discrete convolution definition, $N_s$ should be $N-N_h+1$. This leads to an overdetermined system of equations, which is consistent (i.e., has one or an infinity of solutions) only if a certain number of equations are linear combinations of the rest of the equations. As the coefficients from $\Hhat$ are random by construction, the present system has high chances of being inconsistent, depending on the channel.

Consequently, we give the \textit{s-signal} more degrees of freedom by imposing that only the valid part of the convolution\footnote{The valid part of a deconvolution represents the samples of the central part of the convolution, where the two input signals overlap entirely. These samples are obtained from the summing of $\min(N_s,N_h)$ non-zero terms.} approaches $\yy$, which implies that $N_s=N+N_h-1$. Eq. \eqref{eq:ckg_deconv_matrix} becomes then an underdetermined system, which can have zero or an infinity of solutions. 

According to the Rouché–Capelli theorem of linear algebra, a system of linear equations with $N_s$ variables has a solution if and only if the rank of its coefficient matrix ($\Hhat$) is equal to the rank of its augmented matrix ([$\Hhat|\yy$]). This is true with a high probability in our case because we consider that IR-UWB channels are spatially uncorrelated, and thus $\Hhat$ and $\yy$ are uncorrelated. This implies that adding the extra column $\yy$ in the coefficient matrix $\Hhat$ will not change the rank. 
\end{proof}

\subsubsection{Deconvolution solutions}
\label{sec:deconv_sol}

From a numerical point of view, we need to find an algorithm that provides a solution to our deconvolution problem and check whether it is suitable in the context of uncertain channel estimates represented by $\Hhat$. Several least squares formulations are presented in the present section and their performance is assessed in Section \ref{sec:results1}.

\begin{solution} A maximum likelihood solution (ML):
\begin{eqnarray}
\label{eq:ml0}
\sig^{ML} &=& \argmin{\sig} || \yy - \mb{\Hhat} \sig ||^2 \\
\label{eq:ml}
\sig^{ML} &=& (\Hhatt \Hhat)^{-1} \Hhatt \yy
\end{eqnarray} 
Despite its simplicity, this solution is shown to be instable with respect to the uncertain channel estimates. 
\end{solution}

\begin{solution}  A maximum a posteriori solution, in which we add a  weighted constraint on the searched \textit{s-signal} (MAP):\footnote{The MAP terminology comes from the right term in Eq. \eqref{eq:map0} which could be understood as a prior over the \textit{s-signal}.}
\begin{eqnarray}
\label{eq:map0}
\sig^{MAP} &=& \argmin{\sig} || \yy - \mb{\Hhat} \sig ||^2 + \lambda || \PP \sig ||^2 \\
\label{eq:map}
\sig^{MAP} &=&(\Hhatt \Hhat + \lambda \PPt \PP)^{-1} \Hhatt \yy
\end{eqnarray}

The instability issue of the ML solution can be efficiently addressed by adding a penalty term. This operation is known as a Tikhonov regularization, where matrix $\mb{P}$ is chosen in order to constrain $\mb{s}$ and $\lambda$ is a real scalar trade-off parameter also referred to as weight in the following. The penalty term, also called prior in a Bayesian setting, enforces the desired characteristics of the optimized \textit{s-signal} (e.g., minimal $l_2$-norm, smoothness, etc.), while the first term keeps the result after convolution close to the target channel $\yy$. In our tests, we have chosen to minimize the norm of the searched signal $\sig$ by using $\mb{P}=\mb{I_{N_s}}$ since we do not have other \textit{a priori} information about it. From an implementation point of view, it would have been preferable to use a smooth prior but the accuracy of this type of prior depends on the sampling frequency and can be false resulting in a high penalty of the data fit. 
\end{solution}

\begin{solution} The aforementioned MAP solution with a cross-validation technique in order to set the weight of the constraint (MAP-CV). 

Cross-validation (CV) methods are statistical tools for model validation, i.e., they are employed to evaluate how well a given model will generalize to unknown variations in the data set. Given a known data set of size $N$ (in our case, the channel estimates $\Hhat$ of size $N\times N_s$ and the target channel $\yy$ of size $N \times 1$), a basic cross-validation procedure splits it randomly into a training set of indexes $\mc{T}$ and a validation set of indexes $\mc{V}$ such that $\mc{V} \cap \mc{T} = \emptyset; \mc{V} \cup \mc{T} = \{1,\dots,N\}$. First, Eq. \eqref{eq:map} is solved using only the training data set. Then, the found solution $\sig_t(\lambda)$, is plugged into the given model for the validation set and the performance is evaluated via the generalization error from Eq. \eqref{eq:ge}. The optimal CV $\lambda$ value in a fixed range $[\lambda_{min}, \lambda_{max}]$ is the one that minimizes the generalization error $\Delta_g(\lambda)$. 
\begin{eqnarray}
\label{eq:ge}
\Delta_g(\lambda) &=& || \Hhat_{v} \sig_t(\lambda) - \yy_{v}|| \\
\lambda^* &=& \argmin{\lambda \in [\lambda_{min}, \lambda_{max}]} \Delta_g(\lambda)
\end{eqnarray} 
\end{solution}

\begin{solution}
An expectation maximization solution, which introduces a Bayesian model for the deconvolution problem and provides a joint estimation of its parameters (EM).

In order to apply Bayesian techniques for jointly estimating the \textit{s-signal} and the prior weight, we create a statistical model corresponding to the regularized deconvolution equation. We identify the known data as the target signal ($\yy$) and the hidden data as the searched \textit{s-signal} ($\sig$). The estimated channel ($\Hhat$) is considered a deterministic fixed quantity in the present model. Thus, we have two equations: one for the data fit, where $\mb{e}$ represents the fitting error with independent samples of mean 0 and variance $\epsilon^2$, and one for the signal prior, which offers an artificial representation of the resulting waveform as a noisy zero-mean process of sample variance $\gamma^2$. By identification with the MAP solution, we get $\epsilon^2 / \gamma^2 = \lambda$. Although the MAP solution expressed in Eq. \eqref{eq:map} depends only on the value of $\lambda$, the present model provides a richer description of the underlying phenomena because the two model parameters ($\epsilon$ and $\gamma$) have a concrete meaning (i.e., $\epsilon$ represents the capacity of the model to fit the known data and $\gamma$ can represent, for example, the energy of the searched signal). 
\begin{eqnarray}
\label{eq:data_attach}
\yy = \Hhat \sig + \mb{e}  \text{, with } \mb{e} & \sim & \mathcal{N}(\mb{0},\epsilon^2 \mb{I_N})\\
\label{eq:prior}
 \PP \sig & \sim & \mathcal{N}(\mb{0},\gamma^2 \mb{I_{N_s}})
\end{eqnarray}

Signal estimation using a model with unknown parameters can be solved by an Expectation Maximization (EM) algorithm. EM has been discovered and used independently in several domains ranging from genetics, statistics (estimation of parameters of mixture distributions) to maximum likelihood image reconstruction and speech recognition (estimation of parameters of Hidden Markov models) \cite{Moon96}, \cite{Zoubin99}, \cite{Dempster77}. This iterative algorithm alternates between two steps: i) E-step solving for the hidden variables ($\sig$) knowing the latest parameter estimates ($\epsilon, \gamma$) and the given data ($\yy$); ii) M-step finding the optimal parameters, knowing the current signal estimation and the given data. The convergence can be proved by showing that the algorithm increases the likelihood at each iteration \cite{Borman09}. Therefore, we can write Bayes' rule for the new model and apply the EM algorithm at each iteration step $i$ as follows:
\begin{eqnarray}
\label{eq:bayes_ckg}
p(\sig|\yy,\epsilon,\gamma)& = & \frac{p(\yy|\sig,\epsilon) \times p(\sig|\gamma)}
{p(\yy|\epsilon,\gamma)} \\
\text{E-step:  } \xi_i(\epsilon,\gamma) & = & \mathbb{E}_{\sig|\yy,\epsilon_{i-1},\gamma_{i-1}}
[\ln p(\yy,\sig|\epsilon,\gamma)] \\
\text{M-step:  } (\epsilon_i,\gamma_i) & = & \underset{(\epsilon,\gamma)}{\arg\max} \text{ } \xi_i(\epsilon,\gamma)
\end{eqnarray}
After the developments detailed in Appendix \ref{app_em}, at iteration $i$ we obtain a mean signal $\musi  = E[\sig]$, its associated covariance $\SIGMAs^i$, and the model parameters $\epsilon_i$ and  $\gamma_i$:
\begin{eqnarray}
\SIGMAs^i & = & (\epsilon_{i-1}^{-2} \Hhatt \Hhat + \gamma_{i-1}^{-2} \PPt \PP)^{-1}\\
\musi & = & \epsilon_{i-1}^{-2} \SIGMAs^i \Hhatt \yy \\
\epsilon_i   & = & \sqrt{\frac{\yyt \yy - 2 \yyt \Hhat \musi + \mathrm{Tr}(\Hhatt \Hhat \SIGMAs^i) + \musti \Hhatt \Hhat \musi}{N}} \\
\gamma_i 		 & = & \sqrt{\frac{\mathrm{Tr}(\PPt \PP \SIGMAs^i) + \musti \PPt \PP \musi}{N_s}} 
\end{eqnarray} 

\end{solution}

\section{Results and discussion}
\label{sec:results}

In this section, we will present the simulation framework (Section \ref{sec:sim_setup}) followed by illustrative and statistical evaluations of the deconvolution solutions (Section \ref{sec:results1}). Finally, in Section \ref{sec:results2}, we compare our key generation protocol to a benchmark method inspired by the literature by analyzing the generated traffic and the bit agreement prior to key generation. Note that the public discussion and the reconciliation phases will not be evaluated herein because the focus of this work is on the acquisition of the physical layer signals necessary for quantization as an alternative to higher level key distribution using the physical layer. Thus, operations that are subsequent to quantization do not influence the comparison of the protocols. 
  
\subsection{Simulation setup}
\label{sec:sim_setup}

After computing the \textit{s-signal} at the sampling resolution defined by $F_s = 10$ GHz, we simulate the pseudo-analog s-waveform $s_{AC}(t)$ by sinc-interpolation in order to obtain a simulation frequency of 100 GHz. The obtained waveform is next convolved with the true multipath channel $h_{AC}(t)$. Our tests have been performed for arbitrary realizations of CM1 channels corresponding to indoor line-of-sight IR-UWB links according to the IEEE 802.15.4a standard \cite{ieee802154a}. The CIR realizations have been normalized in energy and only the excess delay part has been considered by ignoring the absolute delays due to the inter-node distances.  

\paragraph{ML and MAP.} 
In terms of implementation, the least squares solutions of Eq. \eqref{eq:ml} and Eq. \eqref{eq:map} are found using \textsc{Matlab}\textregistered's linear least squares solver \textit{mldivide}, which employs the QR factorization and provides a solution $\mb{s}$ with the fewest possible non-zero components. 

\paragraph{MAP and MAP-CV.} 
As explained previously, we choose $\mb{P}=\mb{I_{N_s}}$ corresponding to a minimization of the signal energy since a prior matrix corresponding to a differential kernel $\mb{[1,-1]}^T$ would be adapted at higher sampling frequencies $F_s$ when the signal can be considered smooth. Regarding the cross validation routine, for each channel configuration we employed 20 random partitions with the common CV-splitting percentages (70\% training data and 30\% validation data). 

\paragraph{EM.}
The number of iterations of the EM algorithm is set to 30 after empirical observation of the convergence mechanism. The values of the initial parameters are arbitrarily set to $\epsilon = \gamma = 1$.

\subsection{Deconvolution performance}
\label{sec:results1}

In this section, we aim to describe the degradation produced by the imperfect channel estimation $\mb{\hat{h}_{CA}} \neq \mb{h_{CA}}$ for each deconvolution solution. To this end, we quantify the performance of the deconvolution solutions with the root mean squared error (RMSE) between the pseudo-analog noiseless received signal $(s_{AC}*h_{AC})(t)$ and the target signal $y_{BA}(t)$. Note that the receiver should perform additional windowing and synchronization operations since the sent \textit{s-signal} is longer than the needed observation window. For our results, we have performed preliminary idealized correlation-based signal-level synchronization in order to separate the aspects related to signal reconstruction from the synchronization issues, which do not fall in to the scope of the paper. 

First, we show an illustrative example using particular channel realizations $(h_{AB}, h_{AC})$ denoted as channel configurations ($\mc{C}$). Then, the aggregated results of RMSE over $5000$ channel configurations are given.

\subsubsection{Illustrative example}
\label{sec:illustrative}

A value of $\text{SNR} = 20$ dB is employed for the received signals $y_{BA}(t),y_{CA}(t)$ before channel estimation ($\mb{\hat{h}_{CA}}$) and \textit{s-signal} computation. No noise is added after the transmission of the \textit{s-signal} because for now we only focus on the reconstruction capability of the non-adjacent channel regardless of the noise conditions on the last involved link. 
Figures \ref{fig:comp_mapem}-\ref{fig:comp_mapcv} show a comparison between the target signal that should be reconstructed at C, $y_{BA}(t)$, and the noiseless received signal at C, $(s_{AC}*h_{AC})(t)$, for several deconvolution solutions: ML, MAP with two different weight parameters $\lambda \in\{1, 0.01\}$, EM initiated with $\epsilon = \gamma=1$, and MAP-CV.  

\begin{figure}[h]
\subfigure[ML solution (RMSE $=0.243$)]
{\label{fig:ex_ml}\includegraphics[scale=0.4]{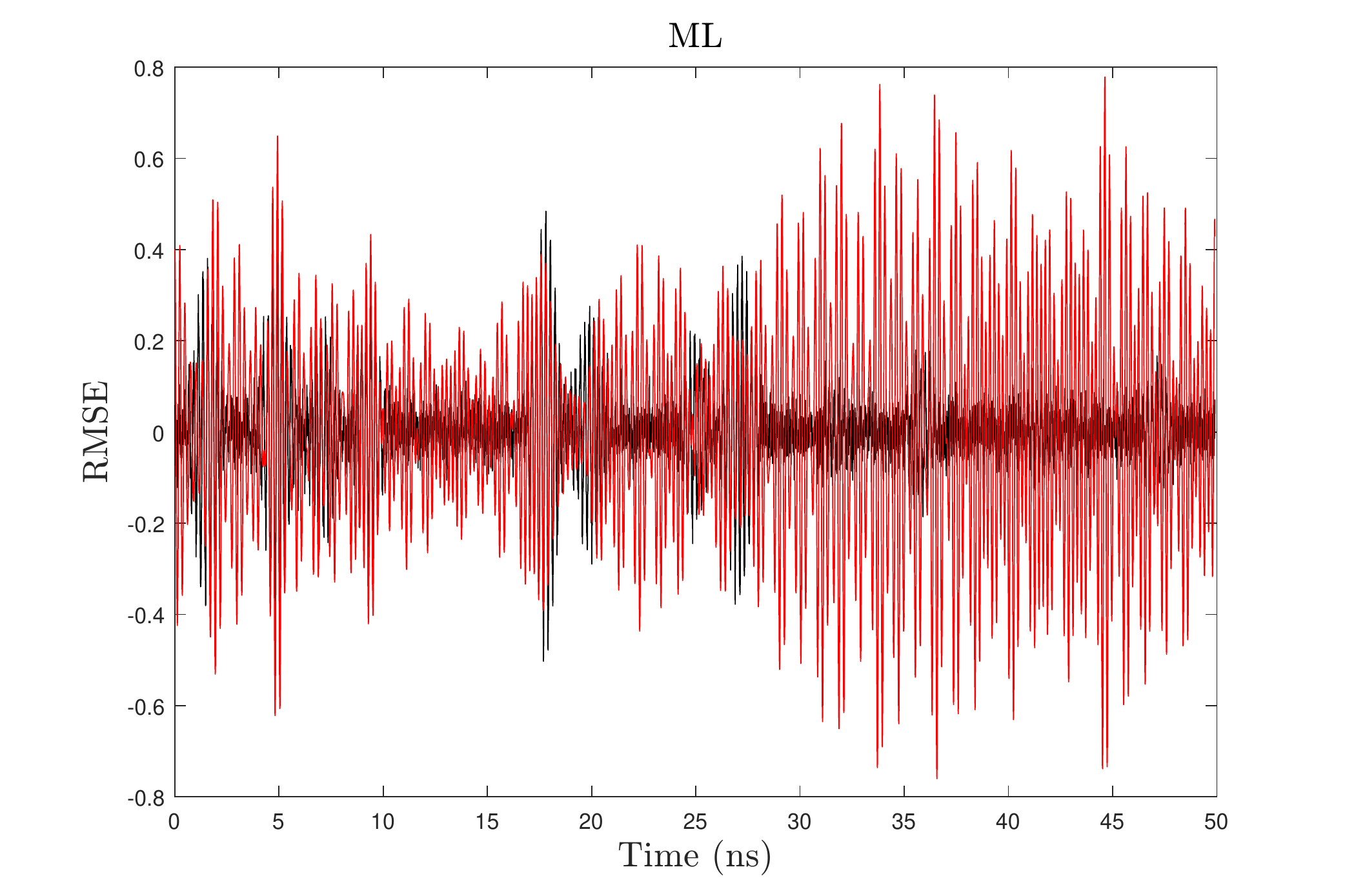}}
\subfigure[MAP solution $\lambda = 1$ (RMSE $=0.067$)]
{\label{fig:ex_map1}\includegraphics[scale=0.4]{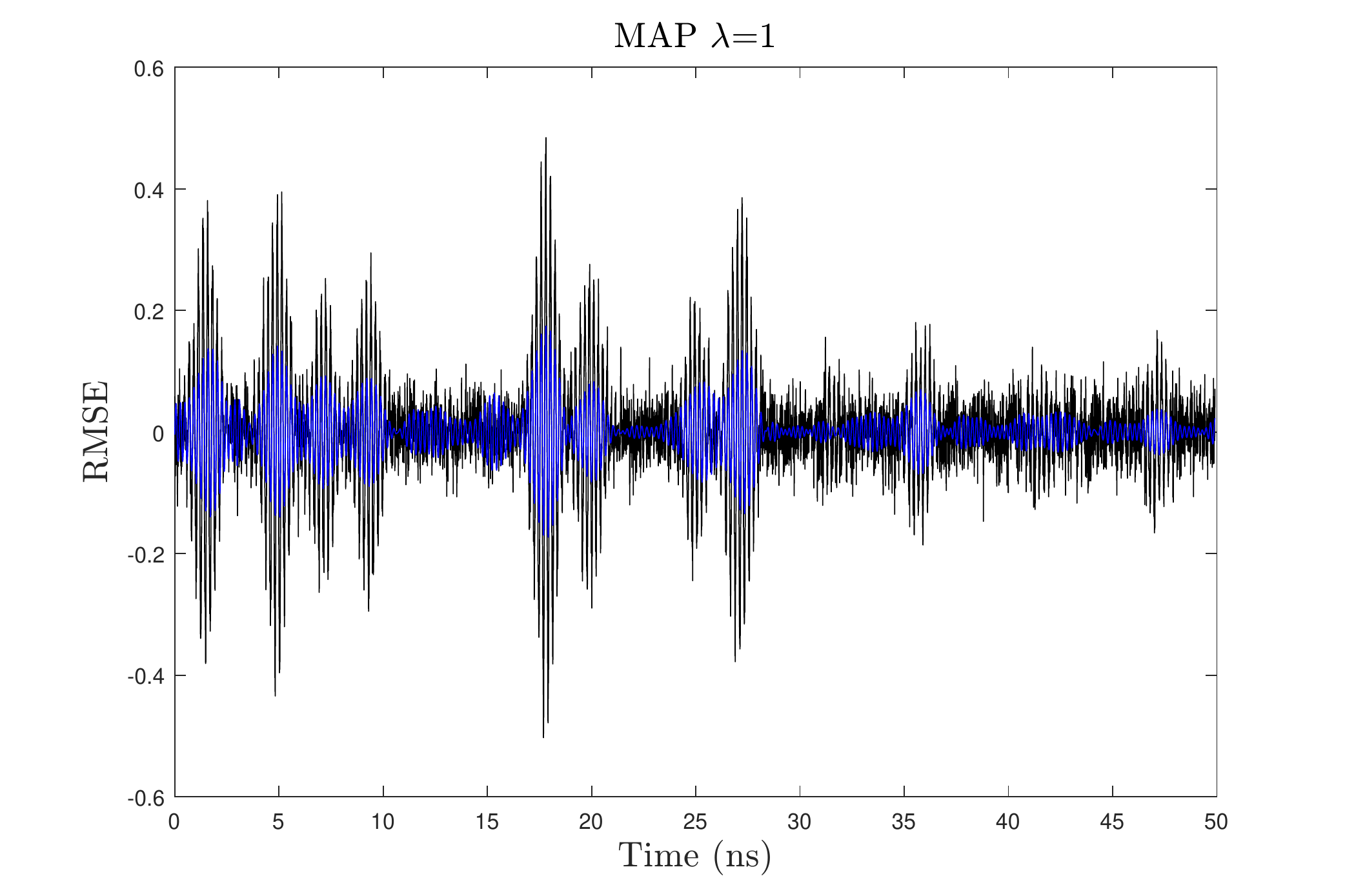}}
\subfigure[MAP solution $\lambda=0.01$ (RMSE $=0.055$)]
{\label{fig:ex_map001}\includegraphics[scale=0.4]{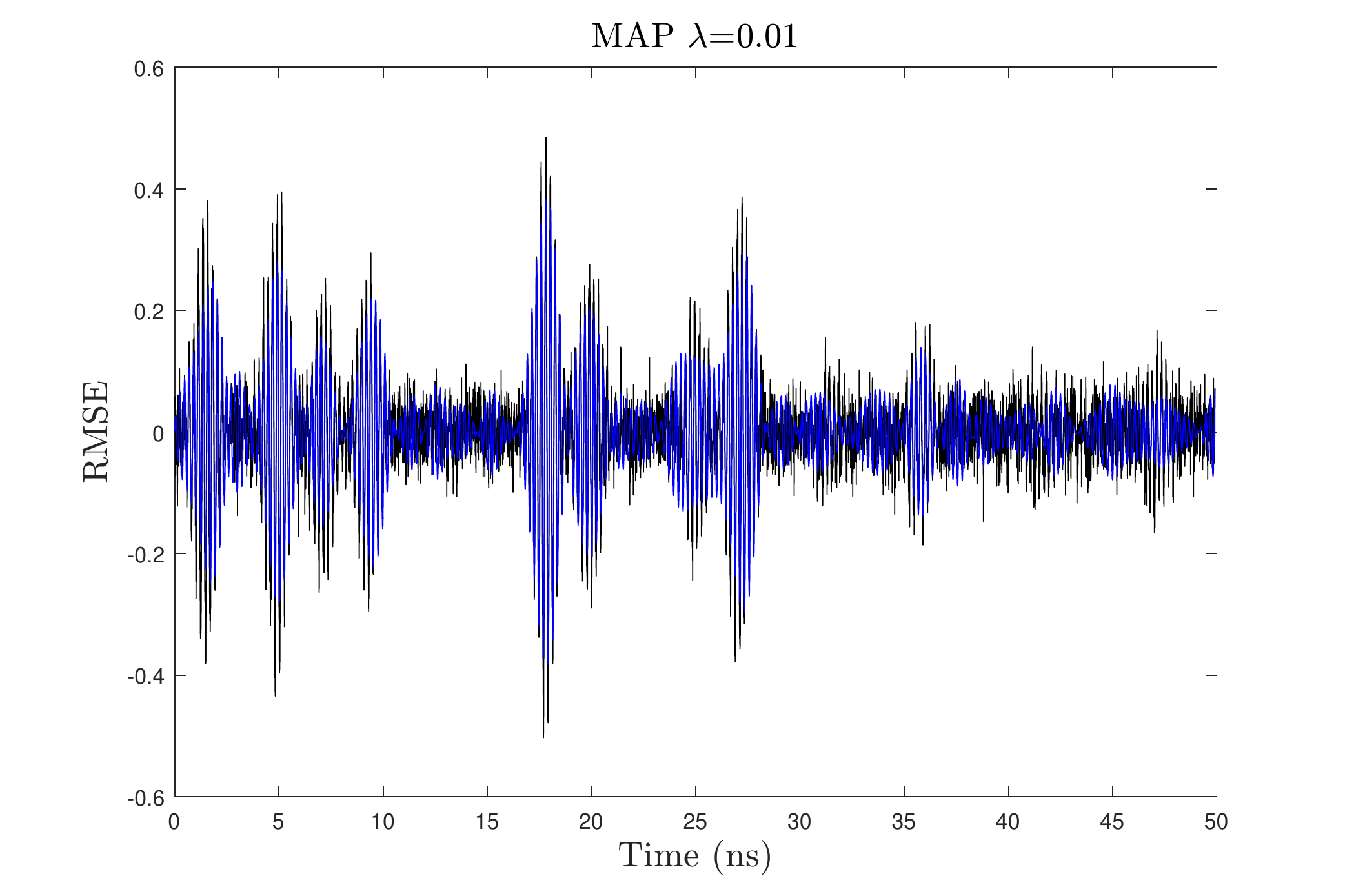}}
\subfigure[EM solution (RMSE $=0.058$)]
{\label{fig:ex_em}\includegraphics[scale=0.4]{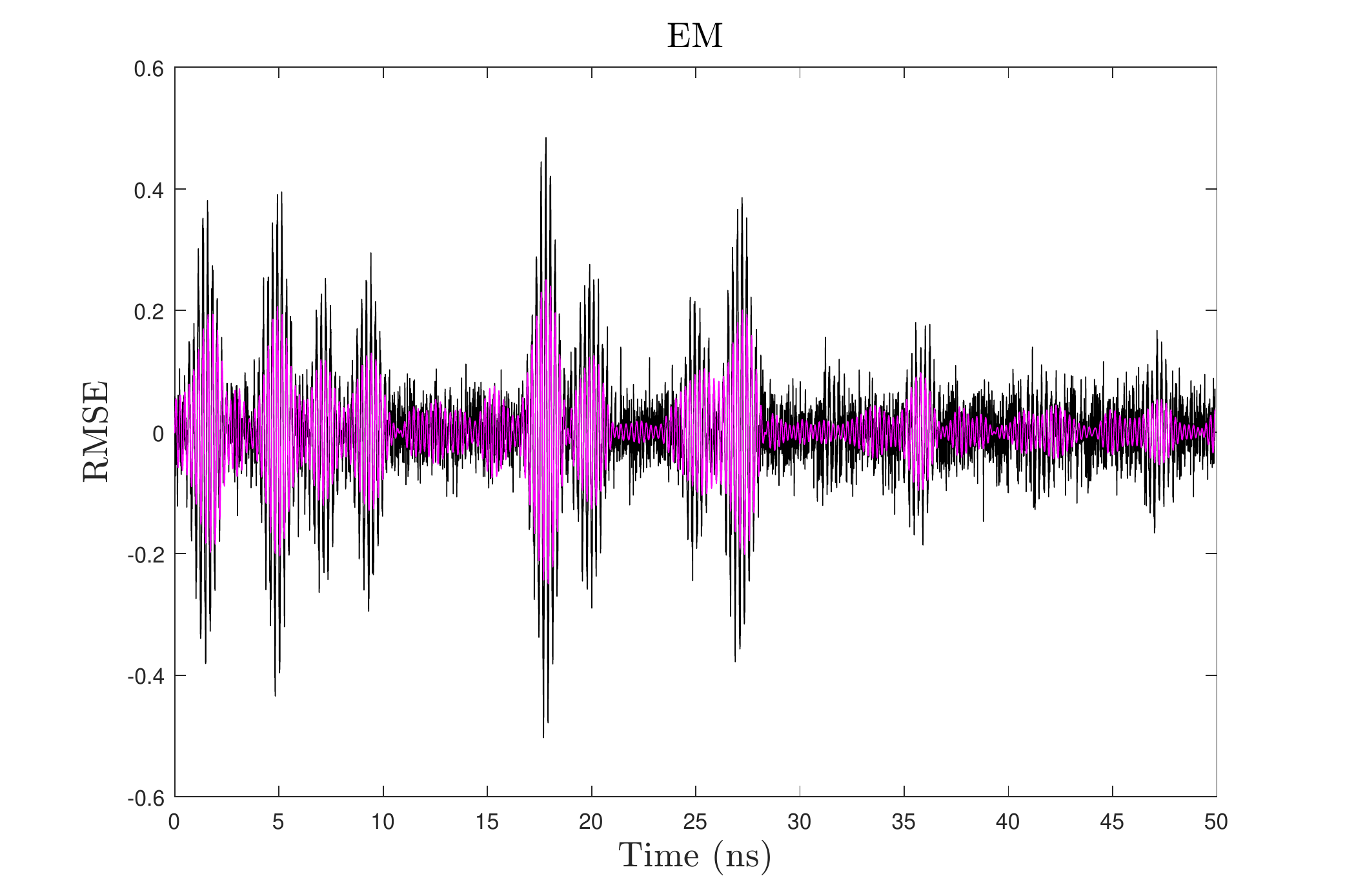}}
\caption{Comparison of various deconvolution solutions for a sample channel configuration $\mc{C}_0$ (target signal in black and received signal in colour)}
\label{fig:comp_mapem}
\end{figure}

\begin{figure}[h]
\subfigure[MAP-CV solution for $\mc{C}_1$ (RMSE $=0.061$)]
{\label{fig:ex_mapcv1}\includegraphics[scale=0.4]{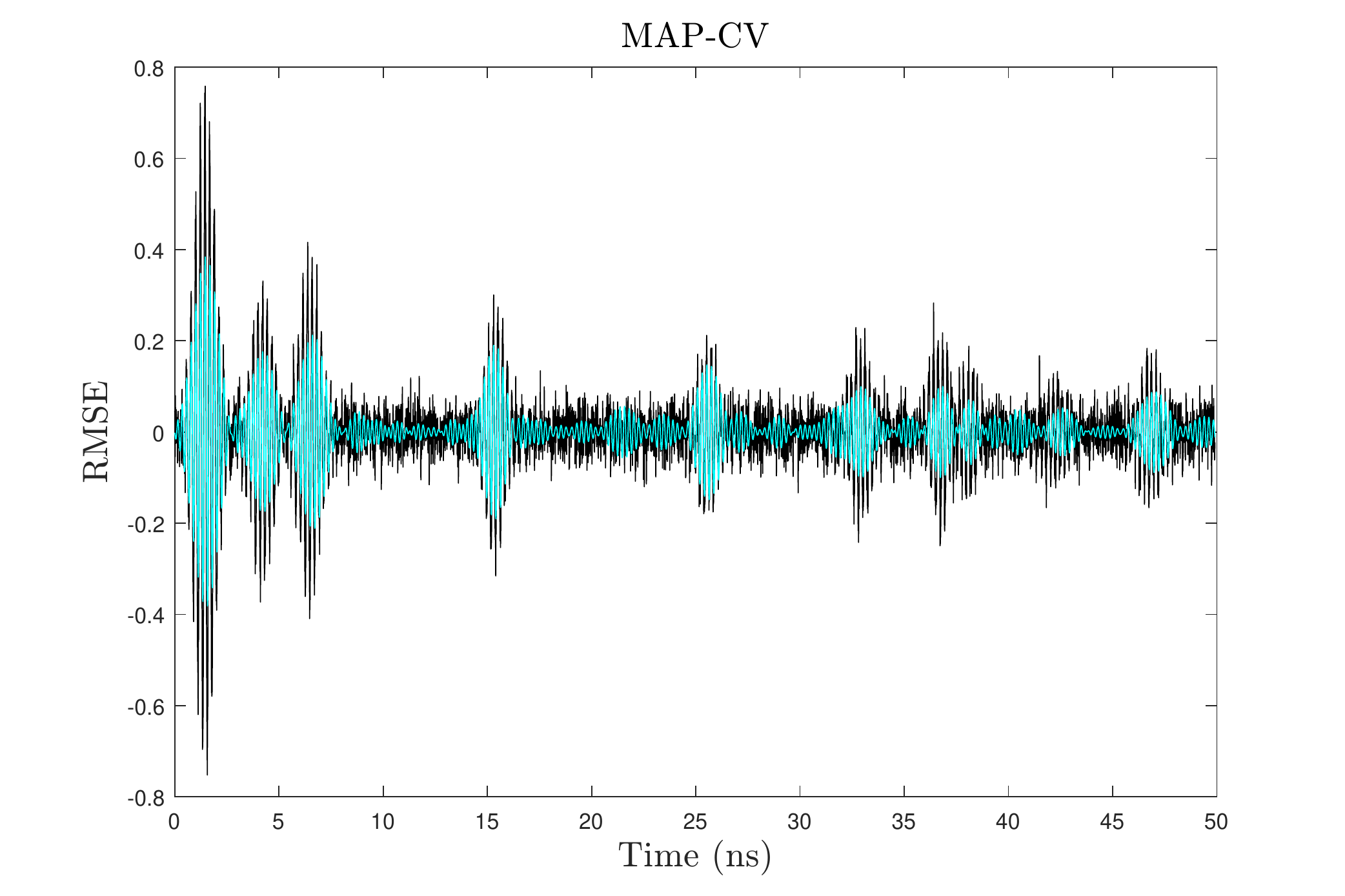}}
\subfigure[MAP-CV solution for $\mc{C}_2$ (RMSE $=0.102$)]
{\label{fig:ex_mapcv2}\includegraphics[scale=0.4]{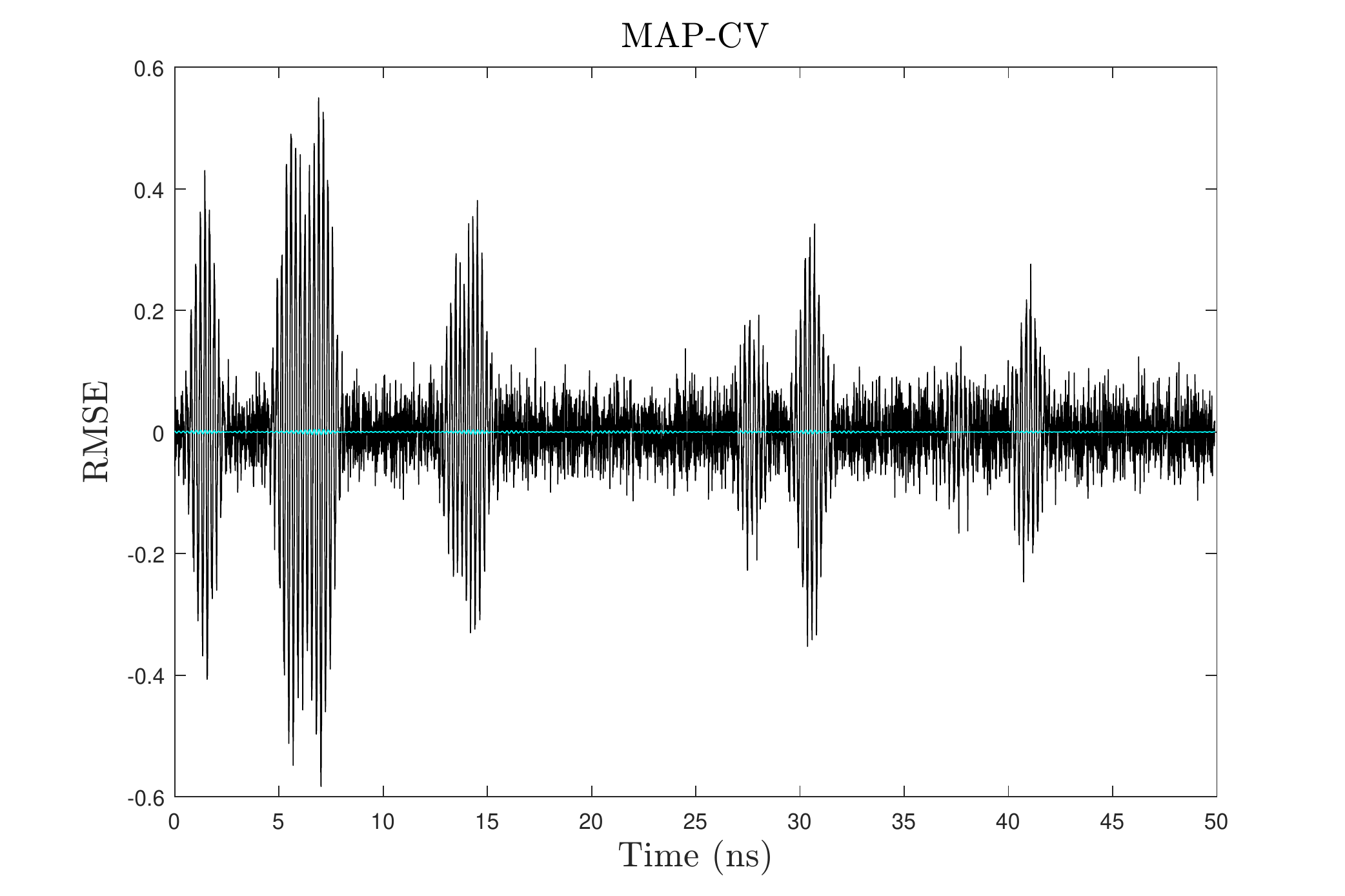}}
\caption{MAP-CV solutions for two sample channel configurations $\mc{C}_1$,$\mc{C}_2$ (target signal in black and received signal in colour)}
\label{fig:comp_mapcv}
\end{figure}

The signal instability observed for the ML solution in Figure \ref{fig:ex_ml} can be regarded as an over-fitting issue: the \textit{s-signal} is computed based on the imperfect channel estimation $\mb{\hat{h}}$ but the final performance depends on the unknown real channel $h(t)$. This means that even though the ML deconvolution solution is exact for the given data (channel estimates and target channel), it could behave unpredictably for slightly different real channel conditions (equivalent to variations in the channel estimates). Although it is not always the case, we have chosen a particular channel configuration in which the ML solution is unstable in order to give an intuitive explanation of the final results.  

Adding the penalty term using the MAP solution avoids the over-fitting. However, the performance is highly dependent on the weight $\lambda$ as shown in Figures \ref{fig:ex_map1}-\ref{fig:ex_map001}. 
We conclude that for this particular channel configuration the regularization works as intended at small $\lambda$ values, but degrades for higher $\lambda$ values, when the prior obtains too much weight and starts flattening the signal ($\mb{s} \rightarrow \mb{0}$).

\begin{figure}[h]
\subfigure[Generalization error for $\mc{C}_1$]
{\label{fig:ex_cvErr1}\includegraphics[scale=0.4]{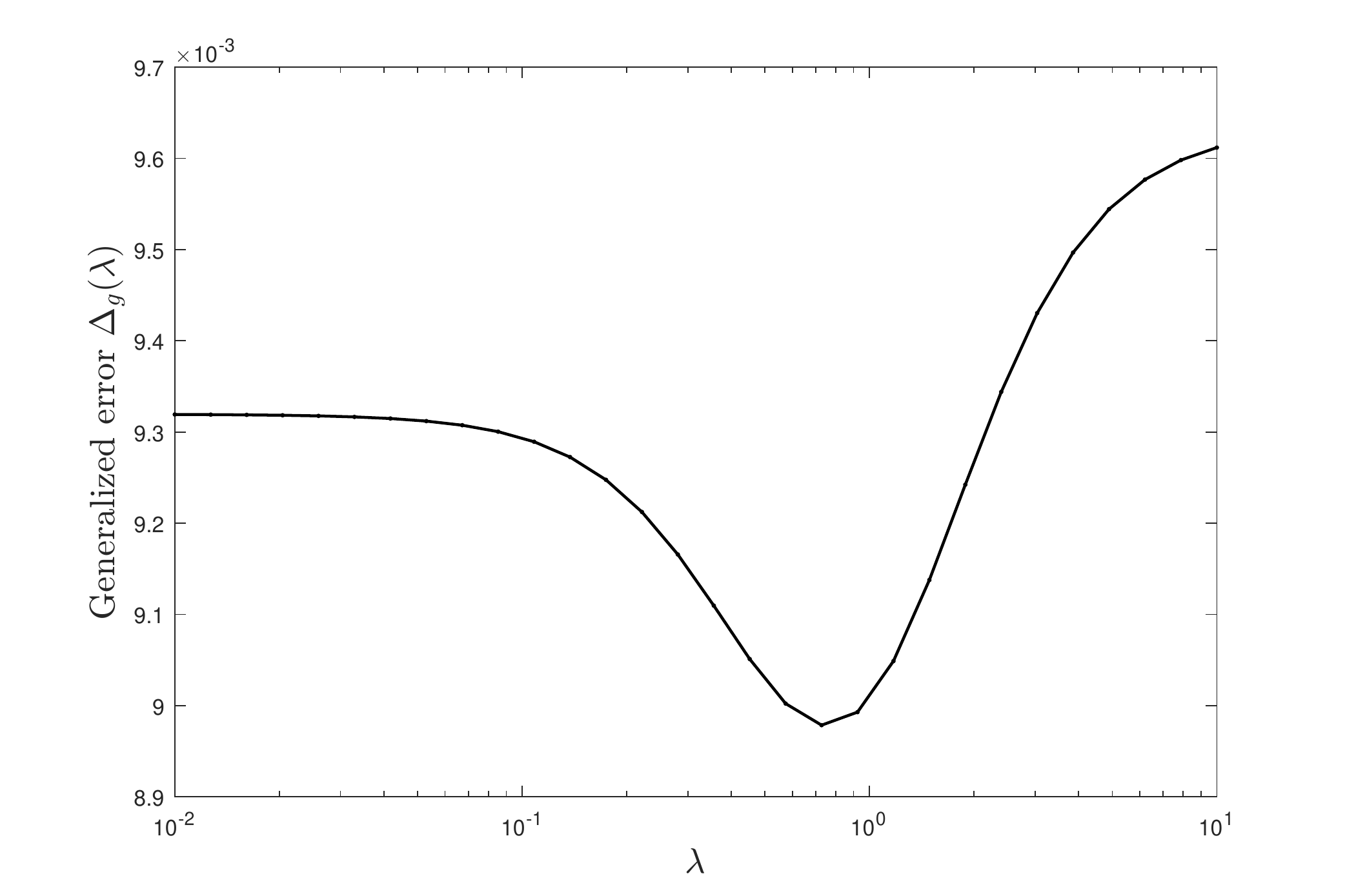}}
\subfigure[Generalization error for $\mc{C}_2$]
{\label{fig:ex_cvErr2}\includegraphics[scale=0.4]{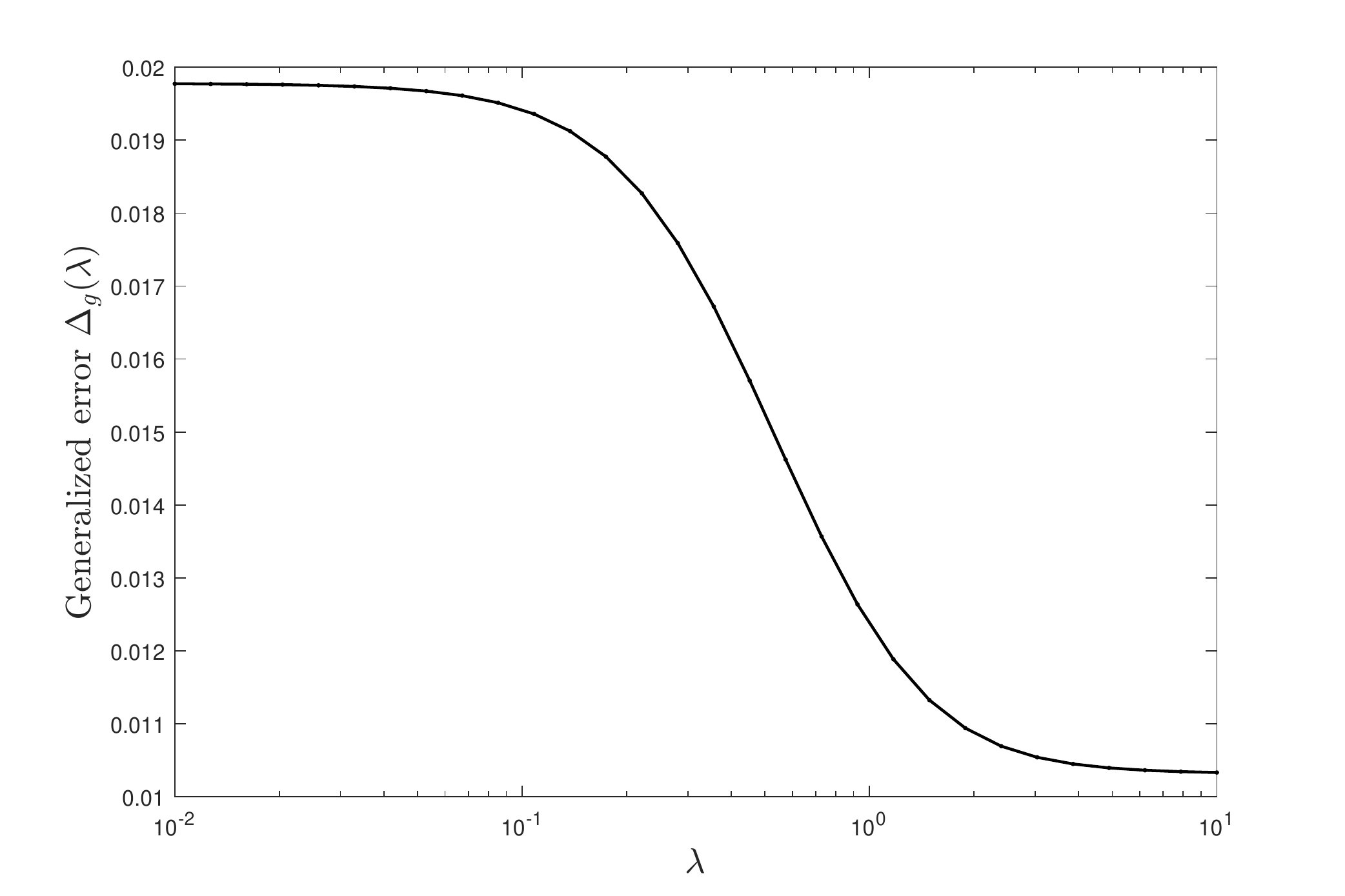}}
\caption{MAP-CV generalization error for two sample channel configurations $\mc{C}_1$,$\mc{C}_2$}
\label{fig:comp_cvErr}
\end{figure}

The solutions that include an automatic parameter setting are EM and MAP-CV. As illustrated for the channel configuration $\mc{C}_0$ in Figure \ref{fig:ex_em}, EM converges to an intermediate solution in terms of RMSE (e.g., equivalent $\lambda = 0.35$ for $\mc{C}_0$). On the contrary, the MAP-CV solution is rather unreliable as shown in Figure \ref{fig:comp_mapcv} for two extreme channel configurations $\mc{C}_1$ and $\mc{C}_2$. The generalization error (Figure \ref{fig:comp_cvErr}) does not always have a minimum and, therefore, model selection using CV does not have optimal performance for all channel configurations. This could be explained by the small size of the data set on which cross-validation is applied when the signals are sampled (500 samples). In certain cases, the validation samples can be insufficient or inadequate to ``generalize'' the estimation of the \textit{s-signal} to the unknown real channel conditions.

\subsubsection{Aggregated results}
\label{sec:aggregated}

First, in Figures \ref{fig:rmse_cdf_analog}-\ref{fig:rmse_cdf_norm} we present the cumulative distribution (CDF) of the RMSE of the unprocessed signals (as illustrated in the previous section) and of the normalized signals with respect to their maximum value. As before, a value of $\text{SNR} = 20$ dB is employed for the received signals $y_{BA}(t),y_{CA}(t)$ before channel estimation ($\mb{\hat{h}_{CA}}$) and \textit{s-signal} computation and no noise is added after the transmission of the \textit{s-signal}. Signal normalization before computing the RMSE is performed in order to look beyond the flattening effect of the regularization. However, the classification of the various deconvolution solutions with respect to their performance is unchanged when normalizing the signals. This confirms the fact that the differences are not just a consequence of the different flattening factors but of the structural differences of the signals issued from each method. The same conclusion is given by the cumulative distribution of the correlation coefficients in Figure \ref{fig:corr_cdf}: the correlation coefficients are used instead of the RMSE for the same signals in order to provide another means of comparison.

\begin{figure}[h]
\subfigure[Cumulative distribution RMSE]
{\label{fig:rmse_cdf_analog}\includegraphics[scale=0.4]{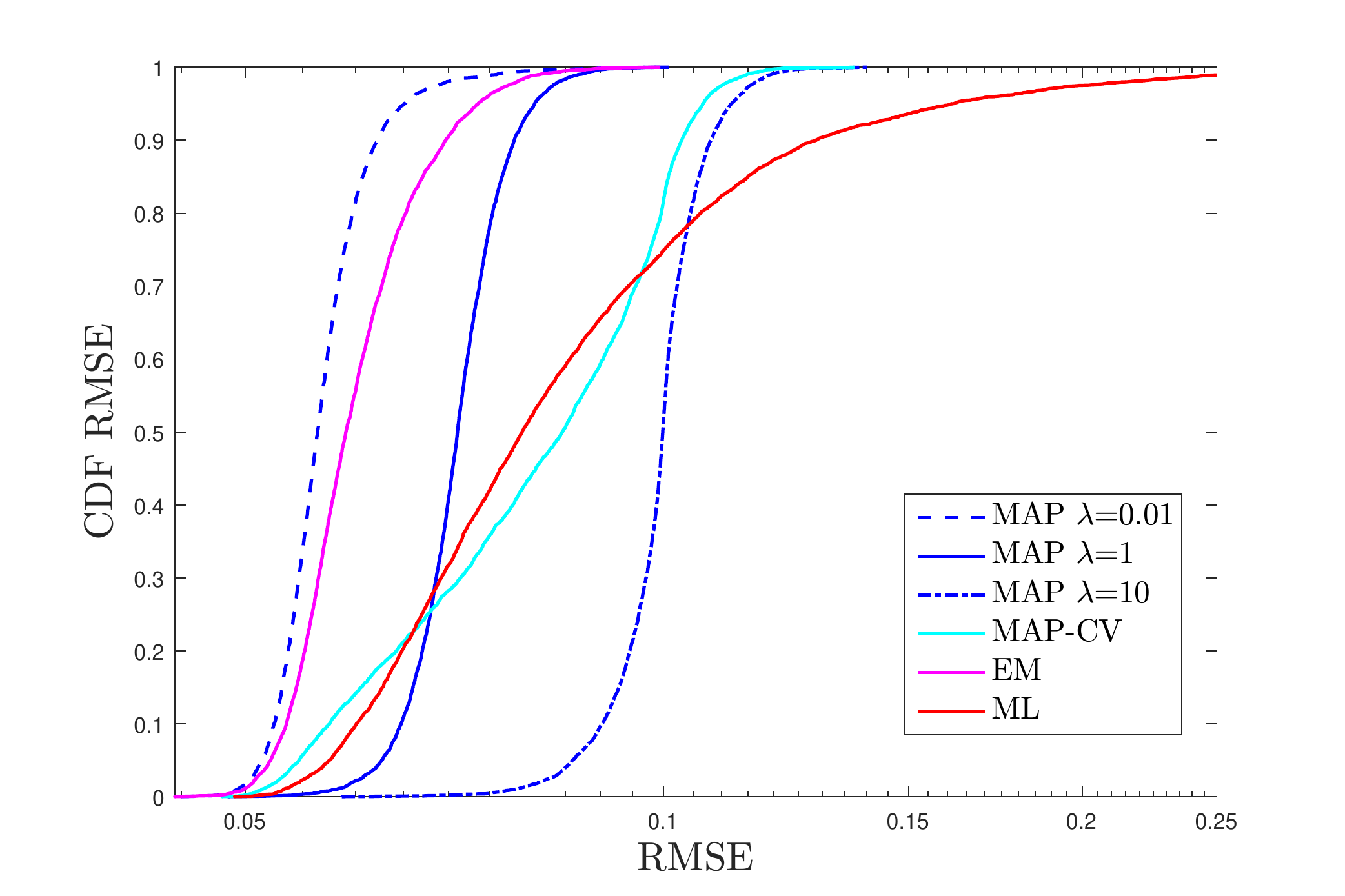}}
\subfigure[Cumulative distribution RMSE (normalized signals)]
{\label{fig:rmse_cdf_norm}\includegraphics[scale=0.4]{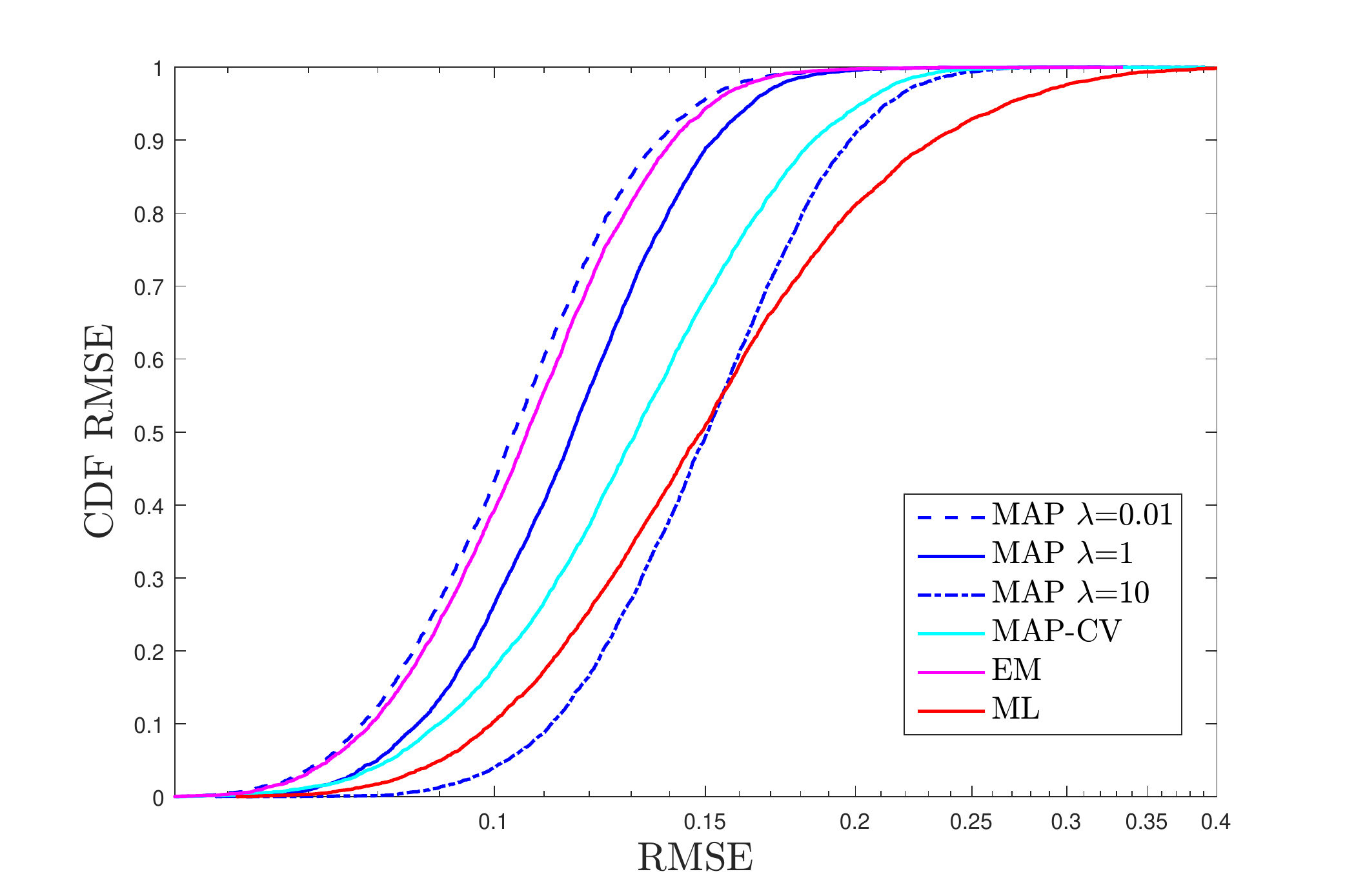}}
\caption{Aggregated RMSE results (noise before estimation only)}
\label{fig:rmse_cdf}
\end{figure}

\begin{figure}[h]
\begin{center}
\includegraphics[scale=0.4]{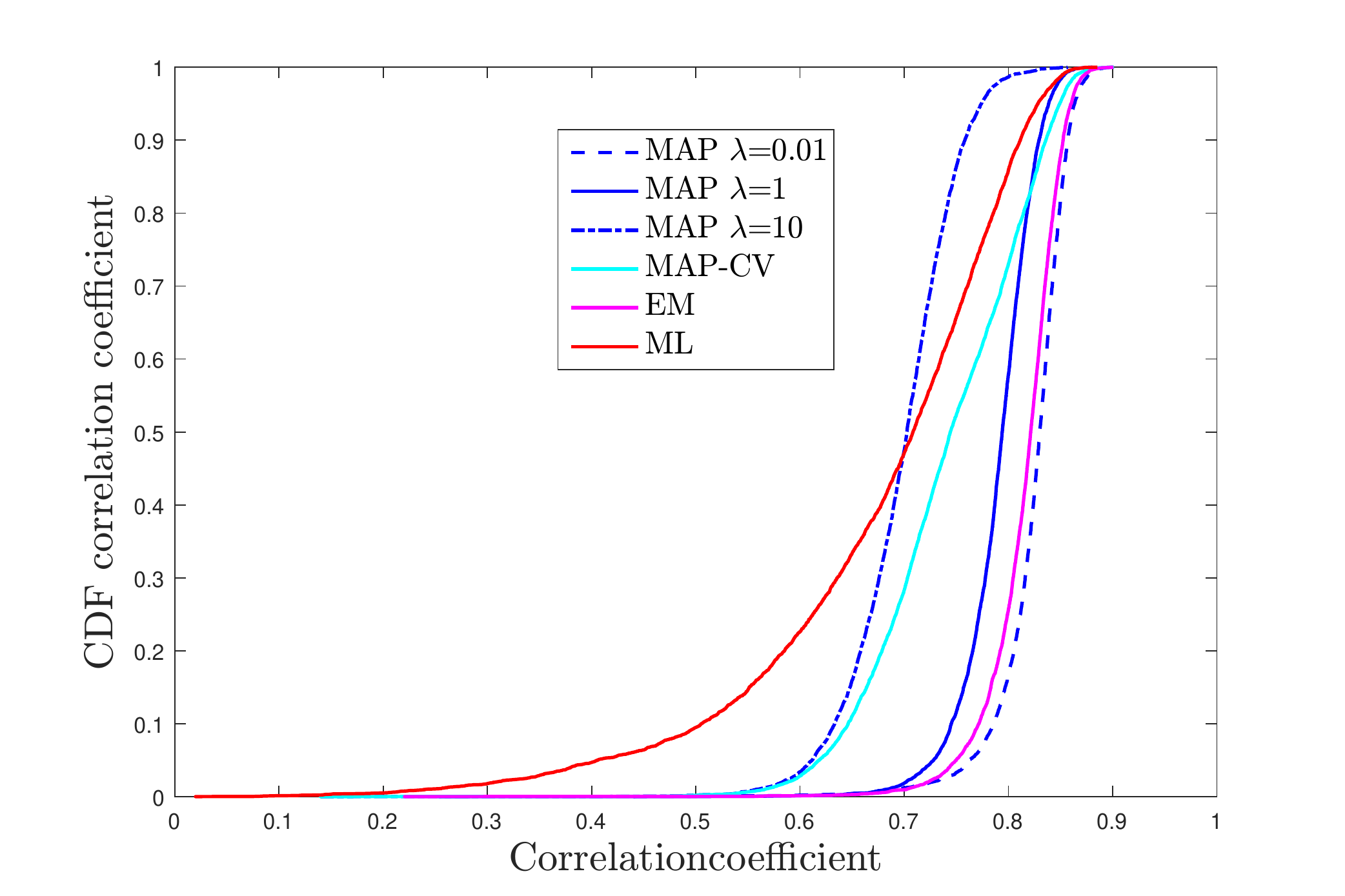}
\caption{Correlation results (noise before estimation only)}
\label{fig:corr_cdf}
\end{center}
\end{figure}

On a relative scale, the ML, MAP-CV and some manually tuned MAP solutions (i.e., at large $\lambda$ values) give the poorest results. It can be also be inferred that at low $\lambda$ values, the MAP estimator has slightly higher performance than the EM solution. However, it is relatively ``unfair'' to compare a solution that gives an estimation of the signal and of the model parameters (EM) with one that takes as an input a favorable manually tuned parameter and provides only an estimation of the \textit{s-signal}. 

Then, in Figure \ref{fig:rmse_cdf_noise}, the results corresponding to entirely noisy scenarios are given. An entirely noisy scenario is defined by adding noise before channel estimation (as before) and after \textit{s-signal} transmission corresponding to a given SNR. As expected the RMSE becomes larger for low SNR values. From the RMSE values strictly it can be inferred that at lower SNR (10 dB), the EM-based solution outperforms the MAP solution with manual parameter $\lambda = 0.01$. However, this is not the case in terms of correlation for lower SNR as it can be seen from Figure \ref{fig:corr_cdf_noise10}. 
 
\begin{figure}[h]
\subfigure[Cumulative distribution RMSE ($\text{SNR}=20$ dB)]
{\label{fig:rmse_cdf_noise20}\includegraphics[scale=0.4]{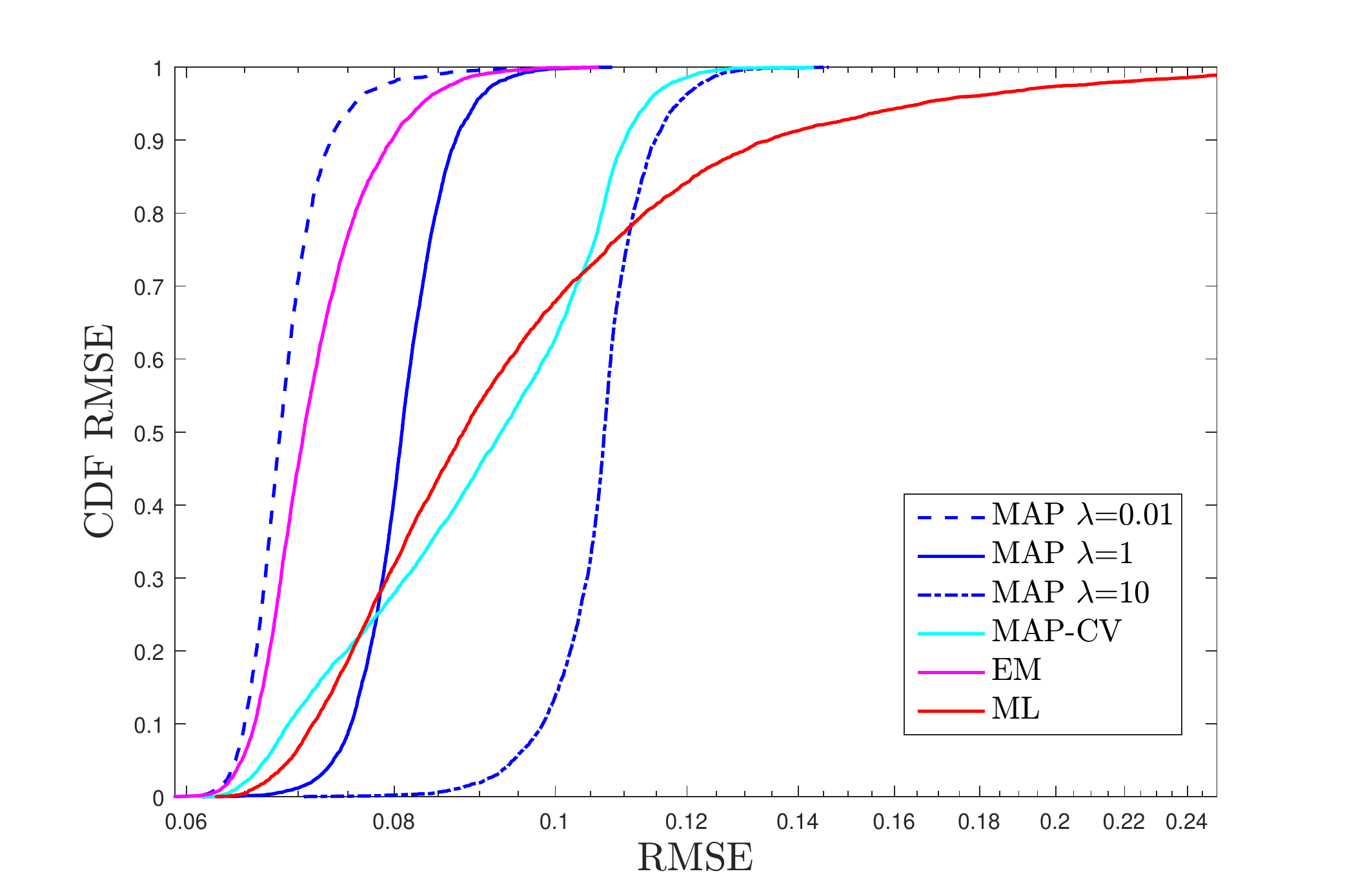}}
\subfigure[Cumulative distribution RMSE ($\text{SNR}=10$ dB)]
{\label{fig:rmse_cdf_noise10}\includegraphics[scale=0.4]{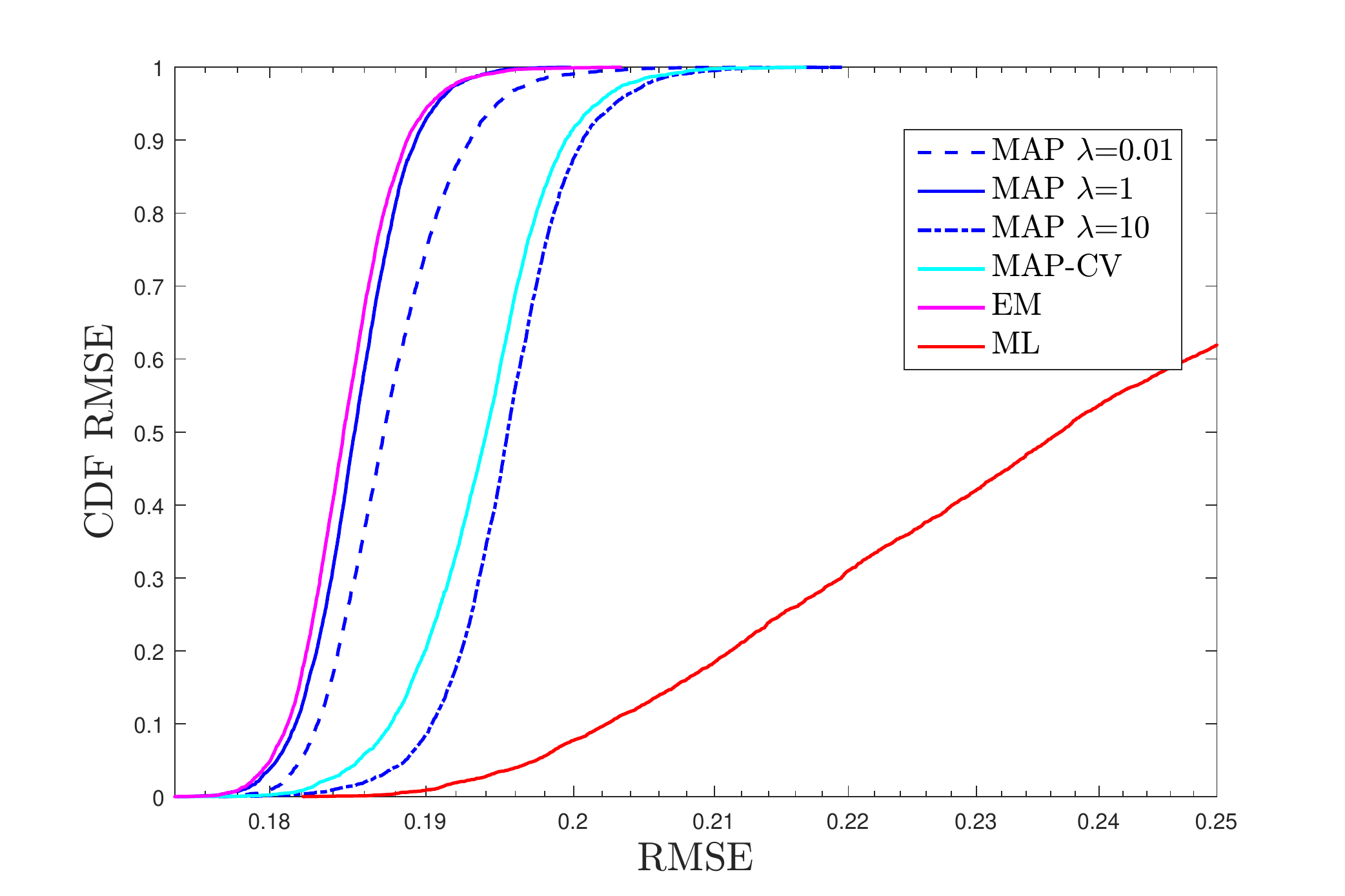}}
\caption{Aggregated RMSE results in entirely noisy scenarios}
\label{fig:rmse_cdf_noise}
\end{figure}

%
\begin{figure}[h]
\subfigure[Correlation coefficients ($\text{SNR}=20$ dB)]
{\label{fig:corr_cdf_noise20} \includegraphics[scale=0.4]{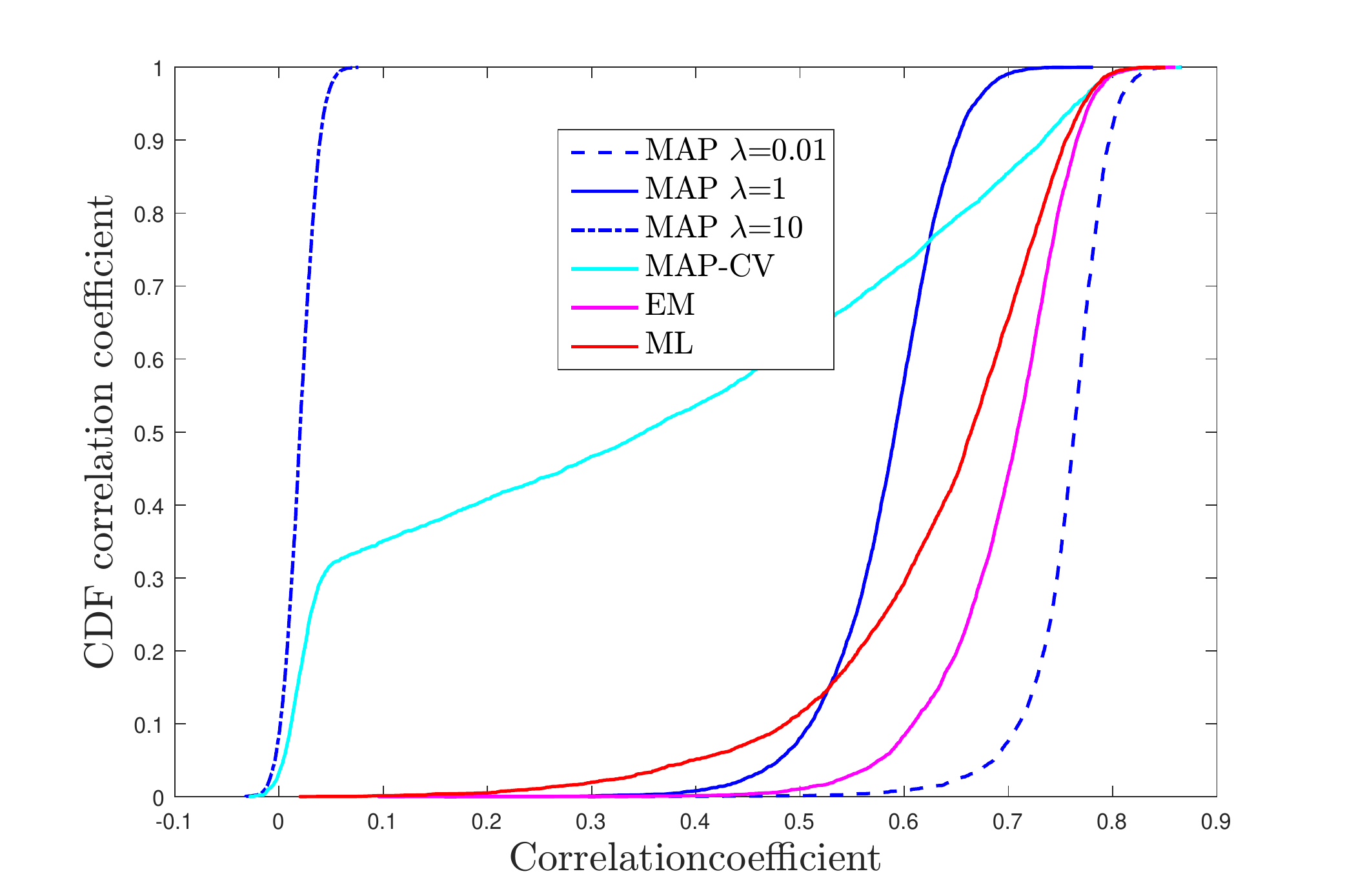}}
\subfigure[Correlation coefficients ($\text{SNR}=10$ dB)]
{\label{fig:corr_cdf_noise10} \includegraphics[scale=0.4]{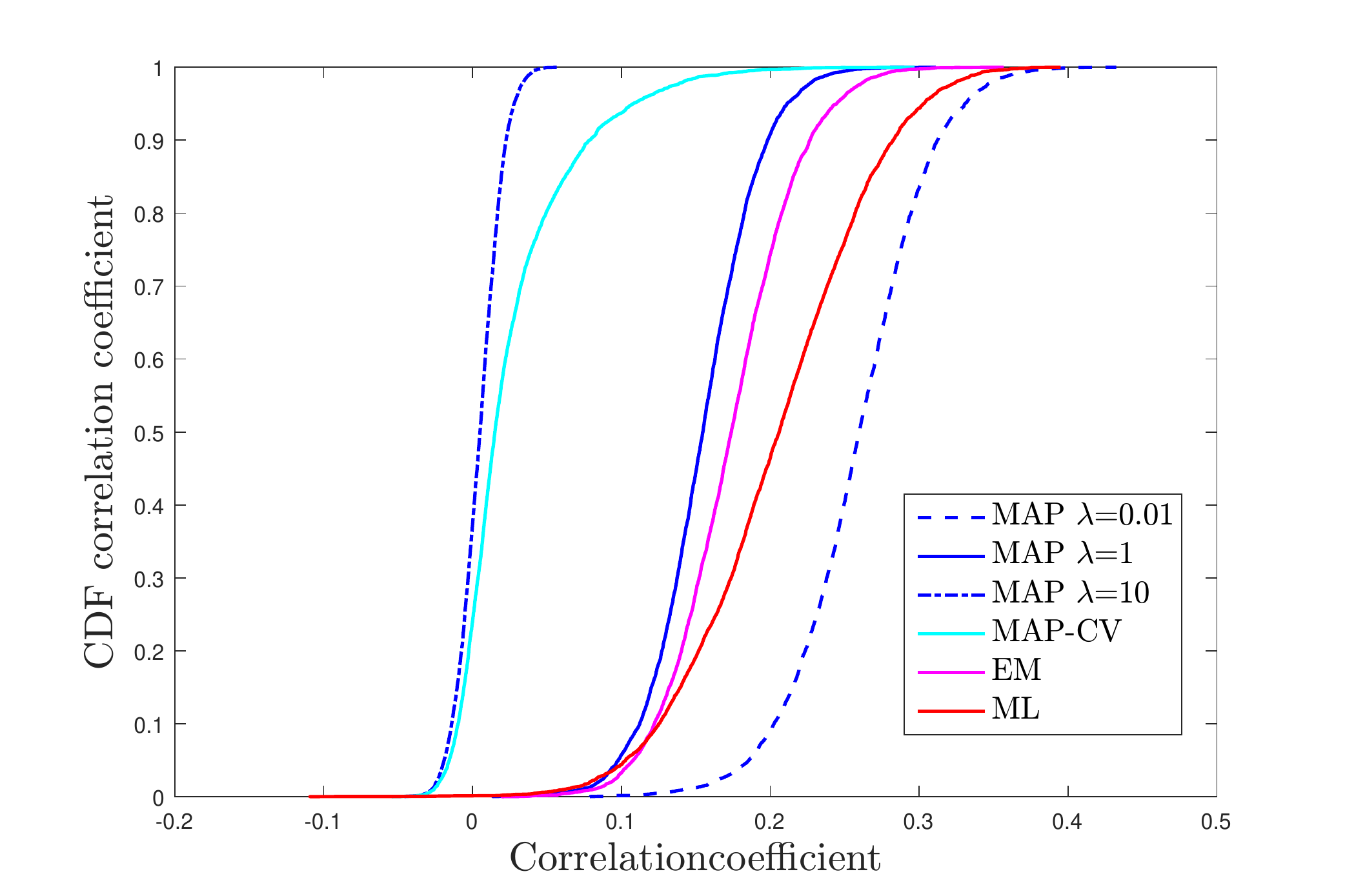}}
\caption{Correlation results in entirely noisy scenarios}
\label{fig:corr_cdf_noise}
\end{figure}

To sum up, the performance achieved by the EM solution with the considered model are comparable or inferior to a manual tuning of the trade-off parameter (i.e., choosing a constant very small prior weight irrespective of the channel configuration). This might suggest that the incurred complexity of the EM solution is unnecessary. However, the EM framework presents several advantages that cannot be achieved by any of the alternative solutions:

\begin{itemize}
\item computation of the uncertainty related to the \textit{s-signal} estimation (see Appendix \ref{app_em}), potentially exploitable by the subsequent quantization layer. 
\item comparison between different prior models in a more complex scenario (e.g., another example of prior could be the information about arbitrary regulatory emission masks compulsory for over-the-air transmissions).
\item joint \textit{s-signal} estimation and channel estimation ($\mb{\hat{h}_{CA}}$) if supplementary information about the channel estimation error model is available. In this case, the hidden variable would be the pair $(\mb{s},\mb{h})$ and the resulting \textit{s-signal} is expected to take into account the uncertainty in the channel estimation that is responsible for the observed artifacts.  
\end{itemize}

\subsection{Protocol performance}
\label{sec:results2}

The proposed physical layer key generation protocol (CKG) is compared to a benchmark physical layer key distribution method (CKD) in terms of exchanged packets (Section \ref{sec:traffic}) or mean key lengths (Section \ref{sec:ba_len}). We also report the mean bit agreement/matching between the acquired bit sequences of the three parties for CKG before reconciliation (Section \ref{sec:ba_len}). The details concerning the quantized signal are presented in Section \ref{sec:qsignal}. 

\subsubsection{Physical layer group key distribution (CKD)}
One way of extending the point-to-point source model for key generation to several links and to a group key is to generate pairwise keys on each link and then distribute a group key (similarly to \cite{Ye07}). CKD can be achieved in two phases.

\begin{itemize}
\item Each node generates a pairwise symmetric key with each of its neighbors based on the properties of the radio channel. The employed key generation method is the same as for our protocol. Contrary to CKG, the key generation for CKD is performed separately for each pair of nodes using only their corresponding channels.    
\item A group key, generated by a lead node using a random number generator, is propagated in the network by XOR-ing operations using the previous single-link keys. The security of the scheme relies therefore on the single-link keys.  
\end{itemize}

\subsubsection{Generated traffic}
\label{sec:traffic}

Table \ref{table:packets_N} compares the number of packets needed for CKD and CKG in a full mesh network of $N$ nodes when CKG is achieved with only one cooperator for each non-adjacent channel. Broadcasting, denoted by $*$ in the table, is used when possible (i.e., pairwise channel probing, dropping and error-correction for CKG). For CKG, a protocol based on the nodes' IDs can be implemented in order to send the necessary \textit{s-signals} such that each node receives all its the non-adjacent channels only once (e.g., each node $i$ sends to a node $j$ the \textit{s-signals} corresponding to the channels $i-k$ with $k>i, k \neq j$). Then, the total number of exchanged packets for \textit{s-signal} transmission is the difference between the total number of channels in a full mesh network $\mylp \frac{N(N-1)}{2}\myrp$ and the number of adjacent channels ($N-1$). The key distribution for CKD in a full mesh consists of $N-1$ packets containing the encrypted group key sent by the lead node to each other node. A lead node is necessary for both CKD and CKG but for different reasons (reconciliation for CKG or distribution for CKD). We consider thus a generic protocol that requires $x$ packets for setting a lead node in a network ($x$ can be 0 if it is established that the lead node is the one with the smallest id).  

Because of the cooperative channel probing phase, CKG is less scalable than CKD for $N \geq 5$. However, larger values of the number of nodes make it difficult to obtain a full mesh scenario and to establish a key during the channel coherence time for both methods and they are thus less practical. 

\begin{table}[ht]
\caption{Exchanged packets for one group key generation/distribution in a full mesh of $N$ nodes}
\centering
    \begin{tabular}{| l | c | c |}
    \hline
    Phase/Method & CKD & CKG\\ \hline \hline
    Pairwise channel probing & $N^*$ & $N^*$ \\ \hline
		Transmission \textit{s-signals} & - &  $N \mylb \frac{N(N-1)}{2} - (N-1) \myrb$ \\ \hline
		Reconciliation : set lead node & - & $x$ \\ \hline
		Reconciliation : dropping & $N(N-1)$ & $N^*$ \\ \hline
    Reconciliation : error-correction & $\frac{N(N-1)}{2}$ & $1^*$ \\ \hline
		Key distribution: set lead node & $x$ & - \\ \hline
		Key distribution & $N-1$ & - \\ 
    \hline
		\hline
		Total & $O(N^2) + x $ & $O(N^3) + x $ \\ 
		\hline
    \end{tabular}
		\label{table:packets_N}
\end{table}

\subsubsection{Quantized signal}
\label{sec:qsignal}

For the following simulations, we consider again $N=3$. The reciprocal CIRs corresponding to the three links are generated independently using the IEEE 802.15.4a statistical model for LOS indoor environments (CM1). The transmitted pulse for initial channel probing $p(t)$ has a bandwidth of 1 GHz (defined at -10 dB of the Power Spectral Density) and a center frequency at 4.5 GHz. The duration of the observation window is set at $T_w=50$ ns. 
For CKG, the sampling frequency for CIR estimation and for the computation of the \textit{s-signals} is set at $F_s = 10$ GHz as before. 
The \textit{s-signal} is computed using the MAP solution with the regularization parameter set manually to $\lambda = 0.01$ dB and the obtained \textit{s-signal} is then filtered to conform with the required signal bandwidth and central frequency. 

The input signal for quantization has a sampling frequency of $F_p = 1/T_p$ and it is normalized with respect to the minimum and the maximum values in order to obtain signal samples between 0 and 1 for all nodes. An example of such a signal issued from link [A-C] and seen in A, B (based on the reception of an \textit{s-signal}), and C is provided in Figure \ref{fig:ex_qSignals}. For simplicity reasons, we choose a two-bit uniform quantization with corresponding cells \{(0-0.25), (0.25-0.5), (0.5-0.75), (0.75-1)\} and a Grey dictionary (\{``00'', ``01'', ``11'', ``10''\}). The employed guard-bands (GB) vary between 0 and 0.1 around the borders of the quantization cells.  

\begin{figure}[h]
\centering
\includegraphics[scale=0.8]{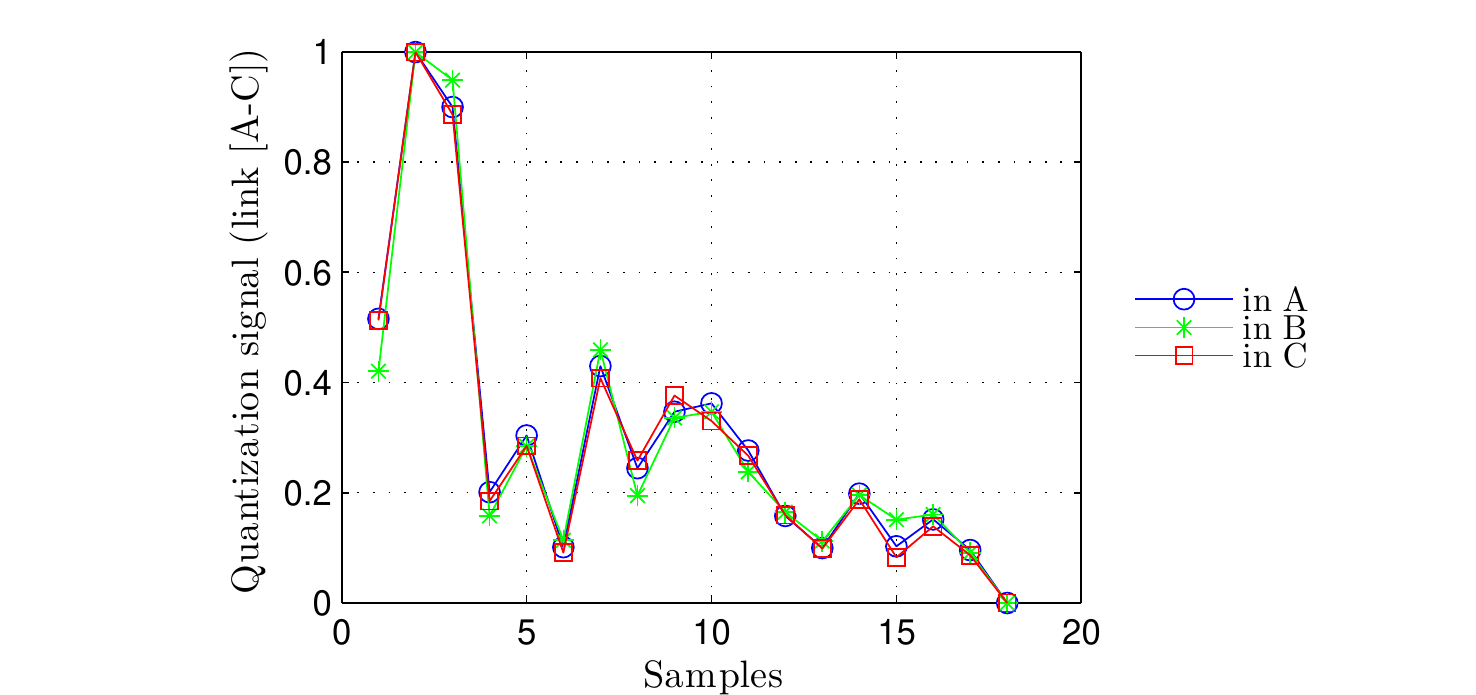}
\caption{Typical input quantization signal seen by the three nodes (SNR$=20$ dB)}
\label{fig:ex_qSignals}
\end{figure}

\subsubsection{Bit generation performance}
\label{sec:ba_len}
The two protocols are compared in terms of the number of generated bits after reconciliation for the same number of over-the-air packets (140 packets corresponding to 10 protocol rounds for CKD and 15 protocol rounds for CKG). As a penalty term, the keys that are not equal after reconciliation are considered to have length 0. 
Averaging is performed over 500 channel configurations, each configuration comprising three channel realizations corresponding to the three pairwise links. As explained in Remark \ref{rem:independency}, the channels are independent from one round to another meaning that the bits can be concatenated without introducing intra-key correlated binary patterns. 

\begin{figure}[h]
\centering
\includegraphics[scale=0.4]{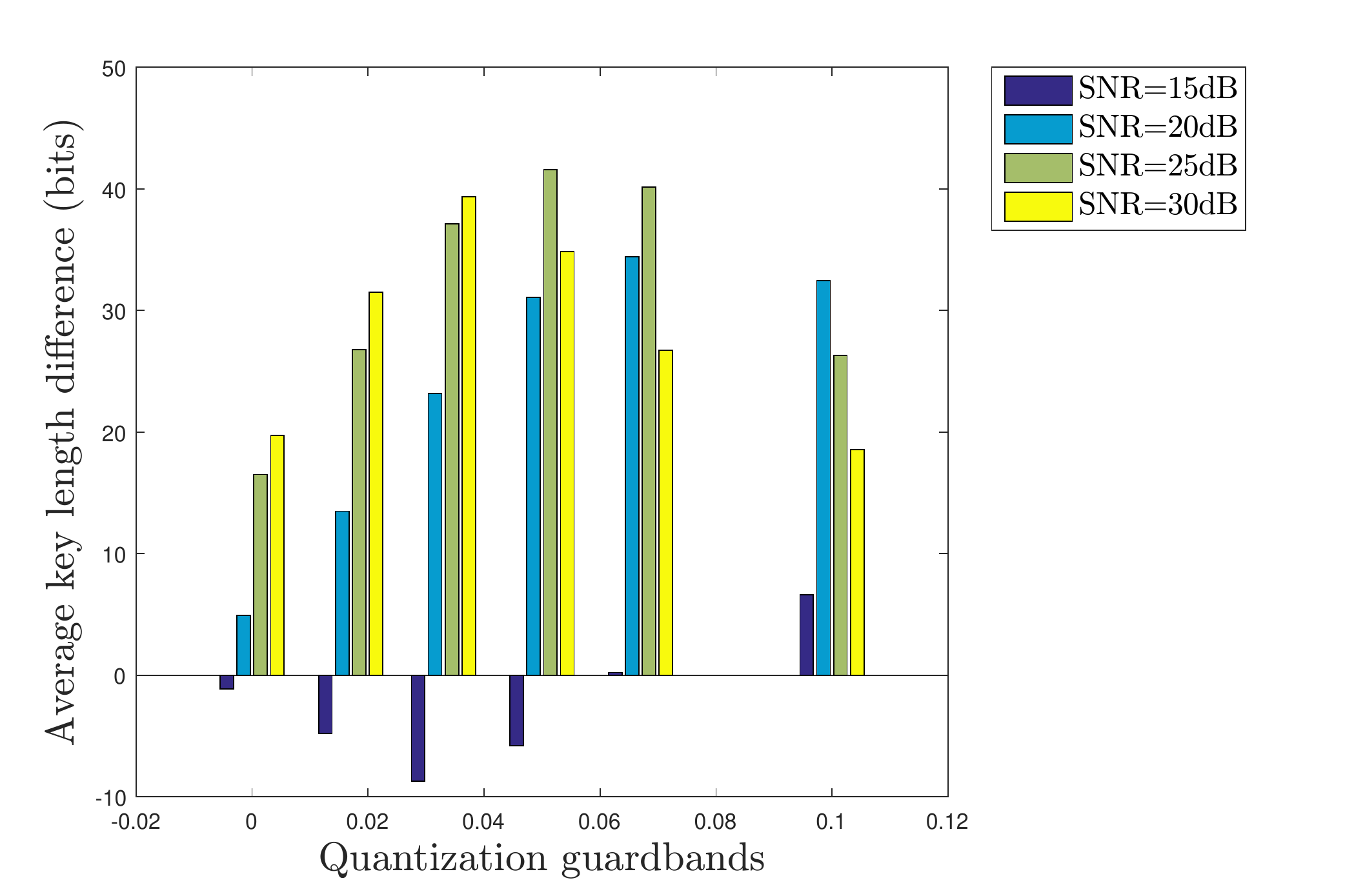}
\caption{Average key length difference CKG-CKD (key gain of CKG w.r.t. CKD)}
\label{fig:difflength}
\end{figure}

\begin{figure}[h]
\subfigure[CKG]
{\label{fig:backg} 
\includegraphics[scale=0.4]{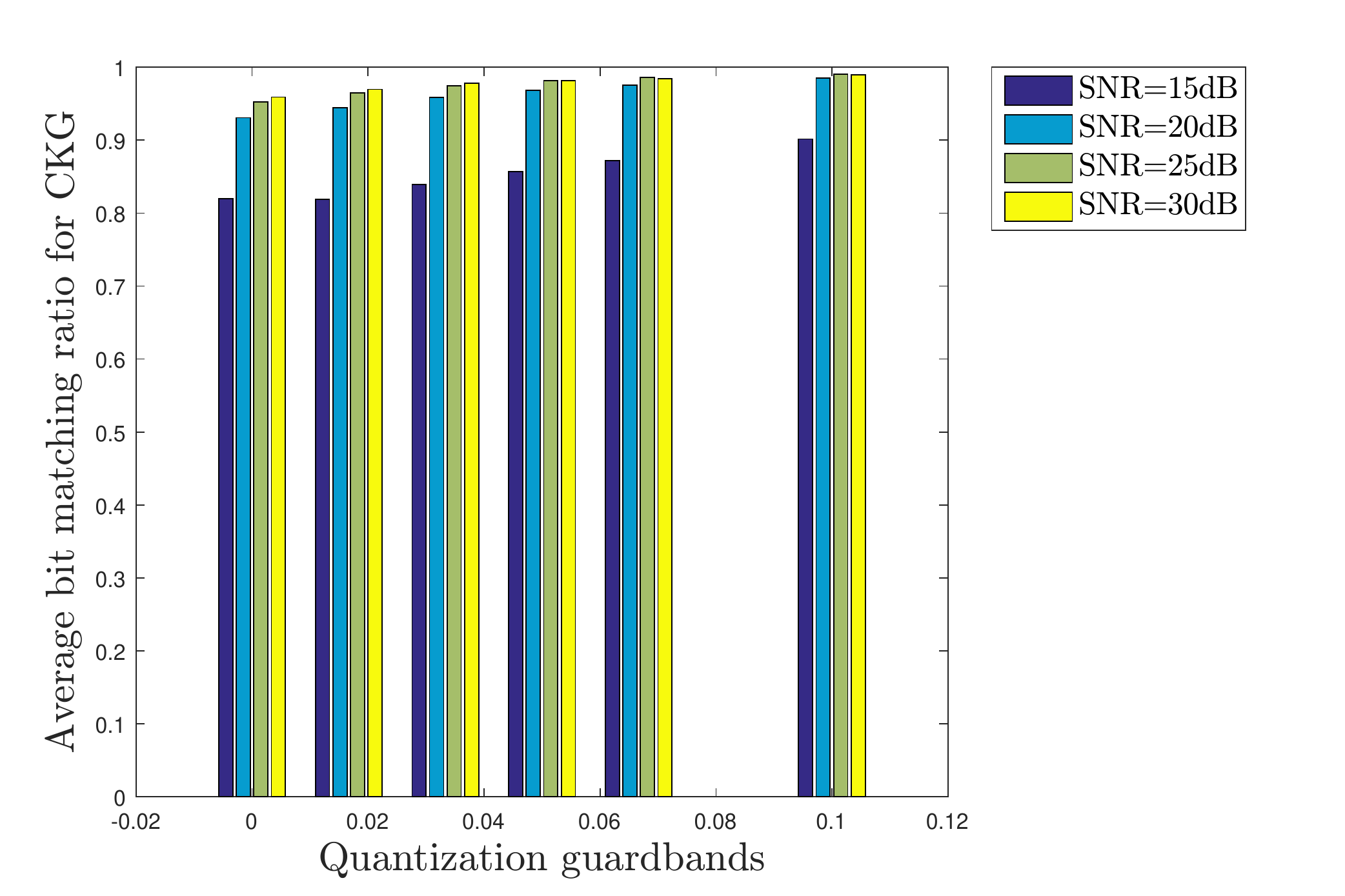}}
\subfigure[CKD]
{\label{fig:backd}
\includegraphics[scale=0.4]{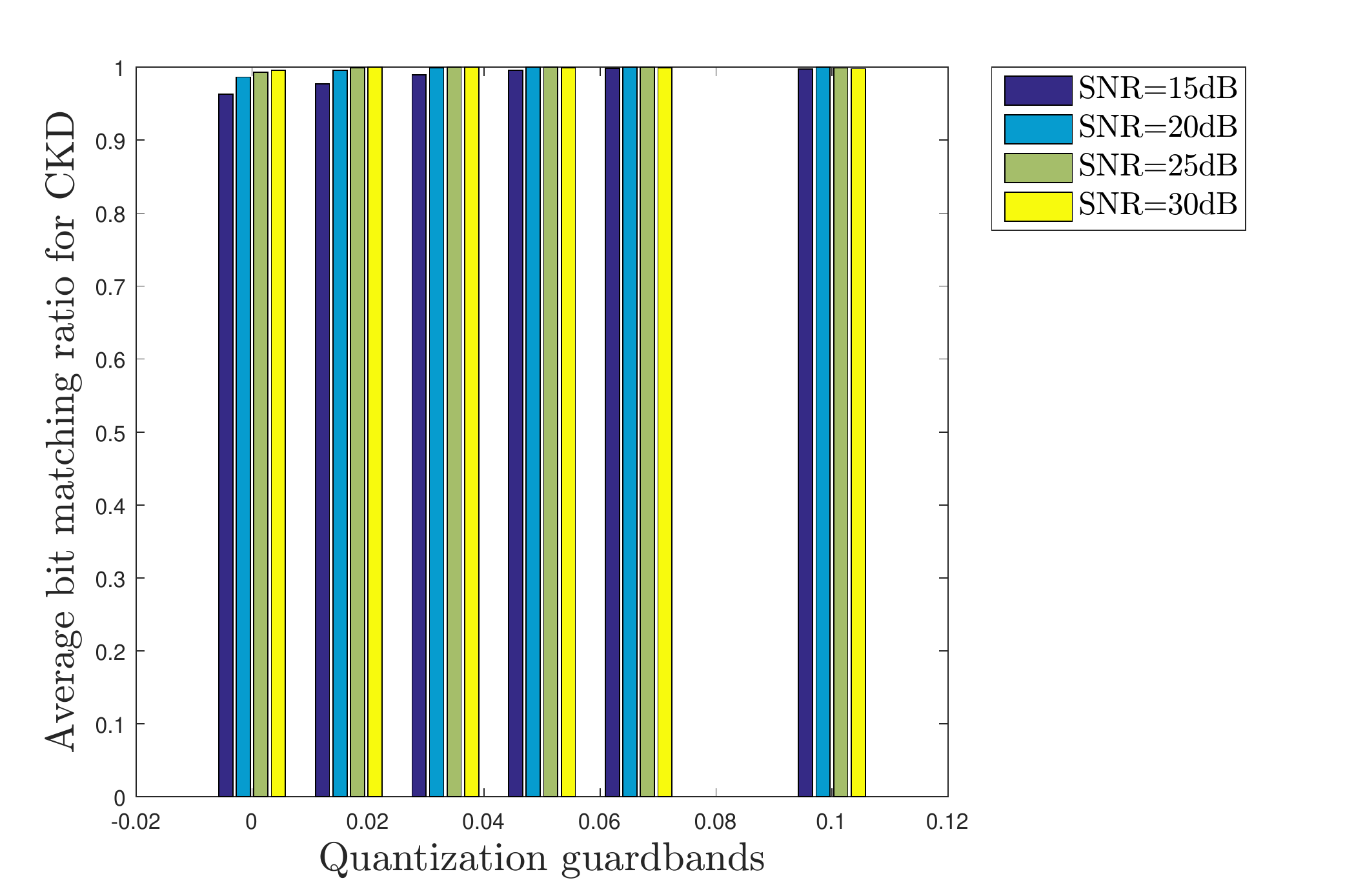}}
\caption{Average bit matching before reconciliation}
\end{figure}

From Figure \ref{fig:difflength}, we conclude that CKG is usually more advantageous in terms of key length, especially at higher SNR. The advantage is maximized at an optimal value of the guard-bands. This is due to the fact that before quantization, the CKG method does not have direct access to the non-adjacent channel measurements that are obtained after a deconvolution operation and another noisy transmission. This degrades the reciprocity of the final samples and leads to more keys that are not agreed upon for CKG at lower SNR or smaller guard-bands. 

However, it can be verified that CKG achieves a relatively high bit matching ratio between the three keys generated at the participating nodes (Figure \ref{fig:backg}). The bit matching ratios for the CKD method (Figure \ref{fig:backd}), although higher than those for CKG, do not present an interest in this study because they can only be computed for pairwise keys and represent solely the single-link reciprocity performance of the system.

\section{Conclusion}
\label{sec:conclusion}

In this paper we have investigated an alternative method for generating secret group keys using the physical layer in IR-UWB systems without relying on classical pairwise key agreement and distribution. For each node, we exploit the concatenation of adjacent and non-adjacent links in a mesh network in order to gather more entropy before the quantization process. 

Our contribution focuses on the parameterized estimation of specific \textit{s-signals} emitted by cooperative nodes in order to \textit{whisper} non-adjacent channel information to their neighbors. We investigate the accuracy of the non-adjacent signal reconstruction phase at the reception side, which is an important factor for subsequent key generation. Consequently, we describe and test two methods for joint maximum \textit{a posteriori} signal estimation and parameter specification. The first one is based on a cross validation technique with the aim of choosing an optimal value for the MAP deconvolution solution, whereas the second one applies the expectation maximization algorithm to obtain a joint estimation of the required signal and of the parameters of the employed statistical model. We conclude that the MAP-CV solution is not robust enough while MAP-EM has stable performance comparable to an advantageous manual tuning method for the MAP parameter. Furthermore, the most attractive feature of the EM model lies in the possibility of extending it to incorporate the channel estimation uncertainty prior to \textit{s-signal} inference. We also analyze the protocol scalability and show the advantages of our method compared to physical layer key distribution in terms of traffic overhead for small cooperative groups and key length for certain parameter configurations. 

Further studies should look into the complexity issues of the \textit{s-signal} generation from two perspectives: software (execution time, band matrix inversions, iterations, etc.) and hardware (programmable transmitters \cite{TCR}, signal dynamics, non-linearity, polarity, etc.). Also, from a key generation point of view, the bit sequences should be further processed by information reconciliation and privacy amplification to obtain final keys. 

The proposed scheme can be adapted to other technologies (narrow-band, OFDM) by changing the cooperative signal generation method according to the most relevant channel features. Moreover, the concept of \textit{s-signal} could be extended to cases where the target signals are not channel estimates but randomly generated information, which leads to a mixed key generation model. However, this complicates the overall key generation procedure because there would be no more direct acquisitions of channel information.

\appendices
\section{The Expectation Maximization solution}
\label{app_em}

For tractability reasons, we consider a Gaussian distribution for the searched signal $\sig$. In order to obtain $\sig$ and the parameters $\epsilon$ and  $\gamma$, the EM algorithm requires two steps for each iteration $i$:
\begin{eqnarray}
\text{E-step:  } \xi_i(\epsilon,\gamma) & = & \mathbb{E}_{\sig|\yy,\epsilon_{i-1},\gamma_{i-1}}
[\ln p(\yy,\sig|\epsilon,\gamma)] \\
\text{M-step:  } (\epsilon_i,\gamma_i) & = & \underset{(\epsilon,\gamma)}{\arg\max} \text{ } \xi_i(\epsilon,\gamma)
\end{eqnarray} 

\paragraph{E-step.} At iteration $i$, the E-step requires two operations:
\begin{itemize}
\item identification of the parameters (mean $\mus$, and covariance $\SIGMAs$) of the searched signal $\sig$ from the expression of the conditional probability density $p(\sig|\yy,\epsilon_{i-1},\gamma_{i-1}$) by factorizing the terms in $\sig$ and those in $\sigt \sig$ from the exponential functions. 
\begin{eqnarray}
p(\sig) &=& p(\sig|\yy,\epsilon_{i-1},\gamma_{i-1}) \\
p(\sig) & = &(2\pi)^{-\frac{N_s}{2}} |{\SIGMAs}|^{-\frac{1}{2}} \mathrm{e}^{- \frac{1}{2} (\sig-\mus)^T \SIGMAs^{-1} (\sig-\mus)} \\
p(\sig|\yy,\epsilon_{i-1},\gamma_{i-1}) & \propto &
p(\mb{y}|\mb{s},\epsilon_{i-1}) \times p(\mb{s}|\gamma_{i-1}) \\
p(\mb{y}|\mb{s},\epsilon_{i-1}) & = &
(2\pi \epsilon_{i-1}^2)^{-\frac{N}{2}} \mathrm{e}^{- \frac{1}{2 \epsilon_{i-1}^2} (\yy-\Hhat \sig)^T (\yy-\Hhat \sig)} \\
p(\mb{s}|\gamma_{i-1}) & = &
(2\pi \gamma_{i-1}^2)^{-\frac{N_s}{2}} \mathrm{e}^{- \frac{1}{2 \gamma_{i-1}^2} (\PP \sig)^T (\PP \sig)}
\end{eqnarray}
\begin{eqnarray}
\sigt \SIGMAs^{-1} \sig - 2 \must \SIGMAs^{-1} \sig + ct = \epsilon_{i-1}^{-2} (\sigt \Hhatt \Hhat \sig -2 \yyt \Hhat \sig + ct) + \gamma_{i-1}^{-2} \sigt \PPt \PP \sig
\end{eqnarray}
with $ct$ a constant. This leads to the expression of the mean of the searched signal $\mus$ and its covariance $\SIGMAs$:
\begin{eqnarray}
\SIGMAs & = & (\epsilon_{i-1}^{-2} \Hhatt \Hhat + \gamma_{i-1}^{-2} \PPt \PP)^{-1}\\
\mus & = & \epsilon_{i-1}^{-2} \SIGMAs \Hhatt \yy
\end{eqnarray} 

\item computation of the expectation based on $\mus$ and $\SIGMAs$ and ignoring the terms that do not depend on $\epsilon$ and $\gamma$ because they do not influence the maximization step. The index of $\mathbb{E}_{\mb{s}|\mb{y},\epsilon_{i-1},\gamma_{i-1}}$ will be neglected for readability purposes.
\begin{eqnarray}
\xi_i(\epsilon,\gamma) & = & ct + \mathbb{E} [\ln p(\mb{y}|\mb{s}, \epsilon,\gamma)] + \mathbb{E} [\ln p(\mb{s} | \epsilon,\gamma)] \\
& = & ct - \frac{1}{2 \epsilon^2} \mathbb{E}[(\yy - \Hhat \sig)^T(\yy - \Hhat \sig)] - \frac{N}{2} \ln (2 \pi \epsilon^2) \nonumber \\
& - & \frac{1}{2 \gamma^2} \mathbb{E}[(\PP \sig)^T(\PP \sig)] - \frac{N_s}{2} \ln (2 \pi \gamma^2) \\
& = & ct - \frac{1}{2 \epsilon^2} [\yyt \yy - 2 \yyt \Hhat \mus + \mathrm{Tr}(\Hhatt \Hhat \SIGMAs) + \must \Hhatt \Hhat \mus] -\frac{N}{2} \ln (2 \pi \epsilon^2)   \nonumber \\
& - & \frac{1}{2 \gamma^2} [\mathrm{Tr}(\PPt \PP \SIGMAs) + \must \PPt \PP \mus] \nonumber - \frac{N_s}{2} \ln (2 \pi \gamma^2)\\
\xi_i(\epsilon,\gamma) &=&  cte  - \frac{1}{2 \epsilon^2} T_1 - \frac{N}{2} \ln (2 \pi \epsilon^2) - \frac{1}{2 \gamma^2} T_2 - \frac{N_s}{2} \ln (2 \pi \gamma^2) \\
T1 &=& \yyt \yy - 2 \yyt \Hhat \mus + \mathrm{Tr}(\Hhatt \Hhat \SIGMAs) + \must \Hhatt \Hhat \mus\\
T2 & = & \mathrm{Tr}(\PPt \PP \SIGMAs) + \must \PPt \PP \mus 
\end{eqnarray}
\end{itemize}

\paragraph{M-step.} At iteration $i$, the M-step consists in the derivation with respect to $\epsilon$ and $\gamma$:
\begin{eqnarray}
T_1 \frac{2\epsilon}{2 \epsilon^4} - \frac{N}{2} \frac{4 \pi \epsilon}{2 \pi \epsilon^2} & = & 0 \\
T_2 \frac{2\gamma}{2 \gamma^4} - \frac{N_s}{2} \frac{4 \pi \gamma}{2 \pi \gamma^2} & = & 0 
\end{eqnarray} 
which leads to the final result:
\begin{eqnarray}
\epsilon_i & = & \sqrt{\frac{T1}{N}} \\
\gamma_i 	& = & \sqrt{\frac{T2}{N_s}} 
\end{eqnarray}


\bibliographystyle{ieeetr}
\bibliography{../my_bib_phd}   

\end{document}